\documentclass[twocolumn]{svjour3}
\smartqed
\usepackage{amsmath}
\usepackage{amssymb}
\usepackage{bm}
\usepackage{graphicx}
\usepackage{xcolor}
\usepackage[colorlinks=true,linkcolor=blue,citecolor=blue,urlcolor=blue]{hyperref}
\newcommand{\hd}{\hat}
\newcommand{\rd}{\mathrm{d}}
\newcommand{\LL}{\mathcal{L}}
\begin{document}
\title{Static electric and magnetic traversable wormholes in\\
$(2+1)$-dimensional nonlinear electrodynamics}
\author{Mauricio Cataldo}
\institute{Mauricio Cataldo \at
Departamento de F\'isica, Facultad de Ciencias,
Universidad del B\'io-B\'io, Avenida Collao 1202, Casilla 5-C,
Concepci\'on, Chile \\
\email{mcataldo@ubiobio.cl}
}
\date{Received: date / Accepted: date}
\maketitle

\begin{abstract}
Every static traversable wormhole reported so far in $(2+1)$ gravity coupled to nonlinear electrodynamics has been obtained from a Lagrangian of the power-Maxwell form $\LL\propto|F|^{k}$ with the single value $k=1/2$, and in every case with a cosmological constant forced to a fixed sign or to zero; the same Lagrangian, again with $\Lambda<0$, underlies the known stationary family. We show that, within the static sector, these restrictions are artefacts of that one choice of Lagrangian, not physical requirements, and give what we argue is a complete classification of static traversable wormholes in this theory. The stationary sector is not addressed here and is not a special case of the present analysis. Of the three mutually exclusive electromagnetic configurations compatible with the symmetry, the radial electric branch admits no throat for any Lagrangian or cosmological constant. In each of the remaining two branches, azimuthal electric and magnetic, the field equations leave only two possibilities: either the redshift function is fixed ($e^{\Phi}=Cr$, or $\Phi=\text{const}$, respectively), in which case the shape function $b(r)$ is completely free and $\Lambda$ is absorbed without constraint; or it is not, in which case $b(r)$ becomes the free function instead, with the redshift determined from it. We solve both regimes in closed or explicit form in both branches, and prove that fixing $\LL(F)$ to any single power $k$, not only $k=1/2$, leaves the shape function determined up to a finite number of integration constants, and in the azimuthal branch forces $\Lambda\neq0$: the freedom found here requires $\LL(F)$ to be genuinely unrestricted. Along the way we identify an azimuthal-electric family whose shape function is exactly the static BTZ mass function, the same geometry that is a black hole in vacuum becoming a traversable wormhole once sourced by the nonlinear field instead. Finally, we apply the classification to the unique conformally invariant power-Maxwell theory in $(2+1)$ dimensions, $k=3/4$, previously known only to source a Coulomb-like charged black hole in the radial branch: the same Lagrangian produces genuine traversable wormholes in the azimuthal and magnetic branches, showing that it is the electromagnetic configuration, not the choice of Lagrangian, that decides between a horizon and a throat; Born-Infeld electrodynamics, by contrast, never violates the null energy condition in any branch and admits no throat at all.
\end{abstract}

\section{Introduction}\label{sec:intro}

Gravity in $2+1$ dimensions has long served as a laboratory in which questions
that are intractable in four dimensions become explicitly solvable. The Weyl
tensor vanishes identically, there are no propagating gravitational degrees of
freedom, and the full curvature is algebraically determined by the Ricci tensor.
The price of this simplicity is that all local geometry must be manufactured by
the matter sector: whatever structure a three-dimensional spacetime possesses is
put there by its sources. The discovery of the BTZ black hole \cite{BTZ,BTZ2}
showed how much structure a negative cosmological constant alone can generate,
and much of the subsequent literature has followed that lead. Electromagnetic
sources have been comparatively less exploited, even though they are the natural
next ingredient and have been known since the work of Deser and Jackiw
\cite{DeserJackiw,DJH} to interact with three-dimensional geometry in ways
without a four-dimensional analogue.

Nonlinear electrodynamics (NLED) is a particularly natural source in this
context. Introduced by Born and Infeld \cite{BornInfeld} on a principle of
finiteness, and later placed on a general footing by Pleba\'nski
\cite{Plebanski}, it has produced the first exact regular black holes in general
relativity \cite{AyonBeato1,AyonBeato2}, regular magnetic black holes and
monopoles \cite{Bronnikov2001}, and, in three dimensions, regular black holes
\cite{CataldoGarcia} and black holes with a Coulomb-like field
\cite{CataldoCruzDelCampoGarcia}. A subfamily that has received sustained
attention is the power-Maxwell class $\LL(F)\propto|F|^{k}$
\cite{HassaineMartinez,HassaineMartinez2,GurtugMazharimousaviHalilsoy}, which
breaks scale invariance except at $k=3/4$ and which, in three dimensions,
possesses features with no four-dimensional counterpart.

Traversable wormholes \cite{MorrisThorne,MTY,Visser} are supported by matter
violating the null energy condition, and the search for physically reasonable
sources with that property has repeatedly turned to NLED: unlike Maxwell theory,
whose stress-energy tensor obeys the NEC identically, a nonlinear Lagrangian can
in principle produce the required effective negative energy densities. In three
dimensions the subject was opened by Perry and Mann \cite{PerryMann}. The
electromagnetic question, however, was answered in the negative by Arellano and
Lobo \cite{ArellanoLobo}, who concluded that NLED with any Lagrangian
$\LL(F)$ cannot support static, spherically symmetric or stationary,
axisymmetric traversable wormholes, in $2+1$ as well as in $3+1$ dimensions.

That conclusion, however, rests on a specific ansatz for the electromagnetic
field rather than on the most general one compatible with the symmetries. The
general form of the field tensor in stationary cyclic $(2+1)$ spacetimes was
determined for Maxwell theory by Ay\'on, Cataldo and Garc\'ia
\cite{Cataldo2002,AyonCataldoGarcia,GarciaDiaz} and extended to arbitrary $\LL(F)$ by
Ca\~nate and Bret\'on \cite{CanateBreton}, whose theorem shows that the field is
governed by three constants $a$, $b$, $c$ subject to $ac=bc=0$, giving two
disjoint branches: one with $c\neq0$, and one in which $a$ and $b$ may be simultaneously nonzero. Translated into the customary three-dimensional language,
in which $F_{tr}$ and $F_{t\varphi}$ are the radial and azimuthal components of
the electric field and $F_{r\varphi}$ is the magnetic pseudoscalar, the azimuthal electric field thus excludes the other two, while in a stationary spacetime a radial electric and a magnetic field can coexist. In the static case, which is the only one considered in this paper, the additional condition $EB=0$ derived in Sec.~\ref{sec:setup} separates the second branch into two, and the
admissible configurations are mutually exclusive: either a radial electric
field, or a magnetic field, or a uniform azimuthal electric field. It is the
last of these that is absent from the ansatz of Ref.~\cite{ArellanoLobo}, and
Ref.~\cite{CanateBreton} accordingly exhibited a stationary counterexample to
the non-existence proof.

Explicit wormhole solutions are by now known in two of the three branches, and
it is worth setting them side by side, because they share a feature that has
gone unremarked. Mazharimousavi, Halilsoy and Gurtug \cite{MHG} found a static
azimuthal-electric wormhole, but only for $\Lambda<0$; their $\Lambda=0$ limit
degenerates into a horizon-bearing configuration. Ca\~nate and Bret\'on
\cite{CanateBreton} found a five-parameter stationary family in the same branch,
again requiring $\Lambda<0$, and showed explicitly that switching off the
cosmological constant destroys the flaring-out condition. Mazharimousavi, Amirabi
and Halilsoy \cite{MAH} found a static magnetic family in which, by contrast,
$\Lambda=0$ is forced by the field equations, and in which the shape function may
be chosen freely; there the redshift function is constant and the coordinate
component of the magnetic field diverges at the throat, precisely the feature
that Ref.~\cite{ArellanoLobo} had invoked against the existence of such
geometries. The common feature is the choice of Lagrangian: all three works take
$\LL\propto\sqrt{|F|}$. That choice turns out to be degenerate, not merely
convenient. For $\LL=\alpha|F|^{k}$ one has
\begin{equation}
\LL-2F\LL_F=(1-2k)\,\LL ,
\label{eq:degeneracy}
\end{equation}
which vanishes identically if and only if $k=1/2$. The combination on the left is
exactly the one that sources the Einstein equation determining the shape function.
At $k=1/2$, therefore, the electromagnetic field is expelled from that equation
altogether, and the geometry of the throat must be produced by whatever else is
available. In the azimuthal electric branch the only remaining source is the
cosmological constant, which is why every static solution in that branch has needed $\Lambda<0$~\cite{MHG}; within the static sector, the necessity is an artefact of the
Lagrangian, not a physical requirement, as shown in Sec.~\ref{sec:electric}. The stationary family of Ref.~\cite{CanateBreton} exhibits the same requirement; whether it has the same origin there is a separate question, which the static analysis presented here does not settle. In the magnetic branch the degeneracy
acts differently, collapsing the nonlinear Maxwell equation to the statement
$e^{\Phi}=\text{const}$ and forcing $\Lambda=0$~\cite{MAH}.

The purpose of this paper is to remove that restriction in the static sector and to give a complete classification of static traversable wormholes in both branches, without fixing $\LL(F)$ to any single power from the outset. The stationary family of Ref.~\cite{CanateBreton}, with nonvanishing angular momentum, lies outside this classification and is not a special case of it; we return to the stationary sector in Sec.~\ref{sec:conclusions}. Our results are as follows. We prove that the field equations admit an arbitrary shape function $b(r)$ if and only if the redshift function takes the form $e^{\Phi}=Cr$, for arbitrary $\Lambda$; in that regime we reconstruct $\LL(F)$ in closed form for any $b(r)$, with the field, $\LL$ and $\LL_F$ finite at the throat, and confined exotic matter in explicit examples. Away from that form, $b(r)$ becomes the free function instead, generically admitting no throat at all except where the redshift function itself becomes singular in a controlled way; we exhibit and fully verify such an exceptional solution. We further show that fixing $\LL(F)$ to a single power $|F|^k$, for any $k$ and not only $k=1/2$, forces $\Lambda\neq0$ and leaves $b(r)$ determined up to a finite number of integration constants, recovering the static solutions of the existing literature as the special, zero-freedom case of this construction; one member of this family reproduces the static BTZ mass function, the vacuum black hole becoming a traversable wormhole once sourced by the nonlinear field instead of vacuum plus $\Lambda$ alone.

The magnetic branch exhibits the analogous structure, and along the way we correct a factor-of-$r$ error in the shape-function--field relation of Ref.~\cite{MAH}. There are again three cases: a regime with $\Phi=\text{const}$ in which any $b(r)$ is admissible but always sources the same one-parameter family of Lagrangians, of which Ref.~\cite{MAH} is the $\Lambda=0$ member; a fixed-power regime, in which restricting $\LL(F)$ to a single power forces the magnetic field itself to a constant rather than a free function of the radial coordinate, the shape function $b(r)$ remains a genuine function of $r$, but with only one free parameter left to choose, $C_1$, rather than being freely specifiable; and a generic regime, exhibited explicitly, in which the roles invert relative to Sec.~\ref{sec:mag-reconstruction}: $b(r)$ is again the free function, chosen first, with the redshift function determined from it instead.

Taken together, these results are exhaustive in the following precise
sense. Proposition~\ref{prop:branches} shows the three electromagnetic
configurations compatible with the symmetry are mutually exclusive, and
Theorem~\ref{thm:radial} excludes the radial branch entirely, for any
$\LL(F)$ and any $\Lambda$. In each of the remaining two branches, every
static solution falls into exactly one of two regimes, distinguished by
Theorems~\ref{thm:rigidity} and~\ref{thm:magnetic-rigidity}: the one with fixed redshift function, where the redshift function is fixed and $b(r)$ is the free function; or the generic one, governed instead by Eq.~\eqref{eq:b-generic} (or its
magnetic analogue \eqref{eq:magnetic-b-generic}), which gives $b(r)$ for
any chosen redshift function and shows, by
Propositions~\ref{prop:generic-nothroat} and
\ref{prop:magnetic-generic-nothroat}, that no throat exists there except at
the boundary of that formula's validity; the explicit solutions of
Secs.~\ref{sec:el-generic} and~\ref{sec:mag-generic} show that boundary is
reached, with the redshift function determined instead once $b(r)$ is
fixed there. No third regime exists, and in both regimes together exactly
one function remains free, together with $\Lambda$. This is what we mean
by a complete classification: not an exhaustive list of solutions, but the
general method, in closed or explicit form, that generates every solution
in either regime, in the same sense that a general solution formula
classifies all solutions of a differential equation rather than
enumerating them one by one.

Beyond this general classification, it is worth asking what the theory
predicts for a Lagrangian singled out on physical rather than
illustrative grounds. The unique power-Maxwell theory with a traceless
stress tensor in $(2+1)$ dimensions, $k=3/4$, was already known to source
a Coulomb-like charged black hole in the radial branch~\cite{CataldoCruzDelCampoGarcia}.
We show in Sec.~\ref{sec:examples} that the same Lagrangian, applied
instead to the azimuthal or magnetic field, produces a genuine
traversable wormhole in each case: the outcome is decided by which branch
of Proposition~\ref{prop:branches} the field occupies, not by the
Lagrangian itself. Born-Infeld electrodynamics~\cite{Cataldo:1999wr}, by
contrast, never violates the null energy condition in any branch and
admits no throat at all, underscoring that the classification's freedom
is not vacuous.

The paper is organised as follows. Section~\ref{sec:setup} fixes the action,
field equations and admissible field configurations. Section~\ref{sec:radial}
treats the radial electric branch, where the null energy condition is saturated
identically and no traversable wormhole is possible for any $\LL(F)$.
Section~\ref{sec:electric} develops the azimuthal electric branch, where
a theorem on the redshift function splits the analysis
into two regimes, each solved in turn. Section~\ref{sec:magnetic} does the same
for the magnetic branch. Section~\ref{sec:examples} applies the
classification to the conformally invariant, Coulomb-like theory.
Section~\ref{sec:conclusions} concludes. We use
geometrized units $G=c=1$ and signature $(-,+,+)$, and we define the
invariant as $F=\tfrac14F^{\mu\nu}F_{\mu\nu}$.

\section{Setup and admissible field configurations}\label{sec:setup}

\subsection{Action and field equations}

We consider $(2+1)$-dimensional general relativity minimally coupled to nonlinear
electrodynamics in the presence of a cosmological constant,
\begin{equation}
S=\int\sqrt{-g}\left[\frac{R-2\Lambda}{16\pi}+\LL(F)\right]\rd^{3}x ,
\label{eq:action}
\end{equation}
where $\LL(F)$ is a gauge-invariant electromagnetic Lagrangian depending on the
single invariant
\begin{equation}
F\equiv\tfrac14F^{\mu\nu}F_{\mu\nu},\qquad
F_{\mu\nu}=\partial_\mu A_\nu-\partial_\nu A_\mu .
\end{equation}
As is customary in three dimensions \cite{CataldoGarcia,ArellanoLobo}, the factor
$1/16\pi$ is retained in order to keep the parallelism with the four-dimensional
theory; since there is no Newtonian limit in $2+1$ dimensions its numerical value
carries no independent meaning. The second invariant
$G\sim{}^{*}F^{\mu\nu}F_{\mu\nu}$ does not exist as a scalar here, so the
restriction to $\LL(F)$ involves no loss of generality.

Variation with respect to the metric gives
\begin{equation}
G_{\mu\nu}+\Lambda g_{\mu\nu}=8\pi T_{\mu\nu},\qquad
T_{\mu\nu}=g_{\mu\nu}\LL-F_{\mu\alpha}F_\nu{}^{\alpha}\LL_F,
\label{eq:einstein}
\end{equation}
with $\LL_F\equiv\rd\LL/\rd F$. Variation with respect to $A_\mu$ gives the
electromagnetic field equation, while $F_{\mu\nu}=\partial_\mu A_\nu-\partial_\nu
A_\mu$ automatically implies the electromagnetic Bianchi identity, unrelated to
the gravitational one $\nabla_\mu G^{\mu\nu}\equiv0$ satisfied identically by
\eqref{eq:einstein}:
\begin{equation}
\big(\LL_F\,F^{\mu\nu}\big)_{;\mu}=0 ,\qquad \rd\bm{F}=0 .
\label{eq:maxwell}
\end{equation}
Throughout we assume $\LL_F\neq0$, the case $\LL_F\equiv0$ corresponding to a
constant Lagrangian, i.e.\ to a mere shift of $\Lambda$.

\subsection{Geometry and traversability}

The static, circularly symmetric line element is written in Morris--Thorne form
\cite{MorrisThorne,PerryMann}
\begin{equation}
\rd s^{2}=-e^{2\Phi(r)}\rd t^{2}+\frac{\rd r^{2}}{1-b(r)/r}+r^{2}\rd\varphi^{2},
\label{eq:metric}
\end{equation}
with $\Phi(r)$ the redshift function and $b(r)$ the shape function. The radial
coordinate ranges from the throat $r_0$, defined by $b(r_0)=r_0$, outwards, and
$2\pi r$ is the proper circumference of a circle centred on the throat.
Traversability requires
\begin{align}
&1-b(r)/r>0\quad\text{for }r>r_0, \label{eq:domain}\\
&\Phi(r)\ \text{finite for all }r\geq r_0, \label{eq:nohorizon}\\
&b(r)-rb'(r)>0, \quad\text{i.e.\ }b'(r_0)<1\ \text{at the throat}.
\label{eq:flareout}
\end{align}
Condition \eqref{eq:nohorizon} excludes event horizons, which would occur where
$e^{2\Phi}\to0$, and \eqref{eq:flareout} is the flaring-out condition deduced
from the embedding of the $t=\text{const}$ slice.

It is convenient to work in the orthonormal frame
\begin{equation}
\bm{e}_{\hd t}=e^{-\Phi}\bm{e}_t,\quad
\bm{e}_{\hd r}=\Big(1-\frac{b}{r}\Big)^{1/2}\bm{e}_r,\quad
\bm{e}_{\hd\varphi}=\frac{1}{r}\bm{e}_\varphi .
\end{equation}
A direct computation gives, for the metric \eqref{eq:metric},
\begin{align}
G_{\hd t\hd t}&=\frac{rb'-b}{2r^{3}} , \label{eq:Gtt}\\[2pt]
G_{\hd r\hd r}&=\Big(1-\frac{b}{r}\Big)\frac{\Phi'}{r} , \label{eq:Grr}\\[2pt]
G_{\hd\varphi\hd\varphi}&=\Big(1-\frac{b}{r}\Big)
\Big[\Phi''+(\Phi')^{2}-\frac{rb'-b}{2r(r-b)}\Phi'\Big] ,
\label{eq:Gpp}
\end{align}
all remaining components vanishing identically,
\begin{equation}
G_{\hd t\hd r}=G_{\hd t\hd\varphi}=G_{\hd r\hd\varphi}=0 .
\label{eq:Goff}
\end{equation}
Equations \eqref{eq:Gtt}--\eqref{eq:Gpp} agree with Eqs.~(10)--(12) of
Ref.~\cite{ArellanoLobo}. The
identities \eqref{eq:Goff} are a property of the ansatz \eqref{eq:metric} alone
and, through \eqref{eq:einstein}, they constrain the matter sector; this is the
mechanism we exploit next.

\subsection{The three branches}

In three dimensions the field tensor has three independent components. In the
orthonormal frame we write
\begin{equation}
E\equiv F_{\hd t\hd r},\qquad
\mathcal{E}\equiv F_{\hd t\hd\varphi},\qquad
B\equiv F_{\hd r\hd\varphi},
\label{eq:components}
\end{equation}
i.e.\ the radial and azimuthal components of the electric field and the magnetic
pseudoscalar. No further components exist. The invariant reads
\begin{equation}
F=\tfrac12\big(B^{2}-E^{2}-\mathcal{E}^{2}\big) ,
\label{eq:invariant}
\end{equation}
so that purely electric configurations have $F<0$ and purely magnetic ones
$F>0$.

Evaluating the stress-energy tensor of \eqref{eq:einstein} in the orthonormal
frame, the diagonal components are
\begin{align}
T_{\hd t\hd t}&=-\LL-\big(E^{2}+\mathcal{E}^{2}\big)\LL_F , \label{eq:Ttt}\\
T_{\hd r\hd r}&=\ \ \LL+\big(E^{2}-B^{2}\big)\LL_F , \label{eq:Trr}\\
T_{\hd\varphi\hd\varphi}&=\ \ \LL+\big(\mathcal{E}^{2}-B^{2}\big)\LL_F ,
\label{eq:Tpp}
\end{align}
Combining the general field equation \eqref{eq:einstein} with the geometric
expressions \eqref{eq:Gtt}--\eqref{eq:Gpp} and the stress-energy components
\eqref{eq:Ttt}--\eqref{eq:Tpp} just obtained gives the fully coupled
Einstein--nonlinear-electrodynamics system for the structural functions $b(r)$,
$\Phi(r)$ and the field content $(E,\mathcal{E},B)$:
\begin{equation}
\frac{rb'-b}{2r^{3}}=\Lambda-8\pi\LL-8\pi\big(E^{2}+\mathcal{E}^{2}\big)\LL_F ,
\label{eq:coupled-tt}
\end{equation}
\begin{multline}
\Big(1-\frac{b}{r}\Big)\frac{\Phi'}{r}\\
=8\pi\LL+8\pi\big(E^{2}-B^{2}\big)\LL_F-\Lambda ,
\label{eq:coupled-rr}
\end{multline}
\begin{multline}
\Big(1-\frac{b}{r}\Big)\Big[\Phi''+(\Phi')^{2}-\frac{rb'-b}{2r(r-b)}\Phi'\Big]\\
=8\pi\LL+8\pi\big(\mathcal{E}^{2}-B^{2}\big)\LL_F-\Lambda ,
\label{eq:coupled-pp}
\end{multline}
together with the electromagnetic field equation \eqref{eq:maxwell}. Equations
\eqref{eq:coupled-tt}--\eqref{eq:coupled-pp} are simply \eqref{eq:einstein}
written out in components; they hold before the branch structure is imposed,
and each of Sections~\ref{sec:radial}--\ref{sec:magnetic} is obtained from this
same system by setting two of the three fields $(E,\mathcal{E},B)$ to zero, as
Proposition~\ref{prop:branches} below requires. Adding
\eqref{eq:coupled-tt} and \eqref{eq:coupled-rr}, for instance, cancels $\Lambda$
and reproduces the first line of \eqref{eq:NECbranches} ahead of its
derivation.

The off-diagonal components of $T_{\hd\mu\hd\nu}$ are
\begin{equation}
T_{\hd t\hd r}=-\mathcal{E}B\,\LL_F,\quad
T_{\hd t\hd\varphi}=E B\,\LL_F,\quad
T_{\hd r\hd\varphi}=E\mathcal{E}\,\LL_F .
\label{eq:Toff}
\end{equation}
Note that $\Lambda$, being proportional to $g_{\mu\nu}$, contributes nothing to
\eqref{eq:Toff}. Imposing \eqref{eq:einstein} together with the geometric
identities \eqref{eq:Goff}, and recalling $\LL_F\neq0$, we obtain the algebraic
system
\begin{equation}
\mathcal{E}B=0,\qquad EB=0,\qquad E\mathcal{E}=0 ,
\label{eq:branches}
\end{equation}
from which the following statement is immediate.

\begin{proposition}[Mutual exclusivity of the branches]
\label{prop:branches}
For a static, circularly symmetric spacetime of the form \eqref{eq:metric} in
general relativity coupled to nonlinear electrodynamics with arbitrary $\LL(F)$
and $\LL_F\neq0$, at most one of the three field components
\eqref{eq:components} may be non-vanishing. The admissible configurations are
therefore
\begin{enumerate}
\item[(i)] the \emph{radial electric} branch, $E\neq0$;
\item[(ii)] the \emph{azimuthal electric} branch, $\mathcal{E}\neq0$;
\item[(iii)] the \emph{magnetic} branch, $B\neq0$.
\end{enumerate}
No dyonic configuration exists in $2+1$ dimensions.
\end{proposition}

Proposition~\ref{prop:branches} has a Maxwell antecedent. For the linear theory,
$\LL\propto F$ with $\LL_F$ constant, the same mutual exclusivity was established
in Ref.~\cite{Cataldo2002}, whose Eqs.~(15)--(17) coincide with our system
\eqref{eq:branches}, and which obtained it in the more general static gauge
$\rd s^{2}=e^{2\alpha}\rd t^{2}-e^{2\beta}\rd r^{2}-e^{2\gamma}\rd\varphi^{2}$
without fixing the areal radius. The Morris--Thorne form \eqref{eq:metric} used
here entails no loss of generality in that respect, since $g_{\varphi\varphi}=r^2$
merely defines $r$ as the areal radius; what Proposition~\ref{prop:branches} adds
is that the conclusion is independent of the constitutive law, holding for
arbitrary $\LL(F)$. In the stationary cyclic case the corresponding statement is
Theorem~1 of Ca\~nate and Bret\'on \cite{CanateBreton}, which generalises to
$\LL(F)$ the Maxwell results of
Refs.~\cite{Cataldo2002,AyonCataldoGarcia,GarciaDiaz}; in their notation the
branches (i), (iii) and (ii) correspond respectively to $b\neq0$, $a\neq0$ and
$c\neq0$, with one difference. In the stationary case $a$ and $b$ need not vanish simultaneously, so that branches (i) and (iii) are not separated there. Their separation is a consequence of staticity, through the second of Eqs.~\eqref{eq:branches}, which follows from $G_{\hd t\hd\varphi}=0$, an identity that fails once the metric carries a nonvanishing shift function. It is therefore specific to the setting of the present paper. It is precisely branch (ii) that is absent from the ansatz of
Ref.~\cite{ArellanoLobo}, whose Eq.~(6) admits only $E$ and $B$; the mutual
exclusivity of those two is recovered here as the second of
Eqs.~\eqref{eq:branches}, in agreement with Eq.~(18) of that reference.

Two consequences of Proposition~\ref{prop:branches} organise the remainder of the
paper. First, since the branches are disjoint, they must be analysed separately,
and a non-existence result in one carries no implication for the others. Second,
the combination that controls the null energy condition,
\begin{equation}
T_{\hd\mu\hd\nu}k^{\hd\mu}k^{\hd\nu}
=T_{\hd t\hd t}+T_{\hd r\hd r}
=-\big(\mathcal{E}^{2}+B^{2}\big)\LL_F ,
\end{equation}
evaluated with $k^{\hd\mu}=(1,\pm1,0)$, reduces in each branch to
\begin{equation}
T_{\hd\mu\hd\nu}k^{\hd\mu}k^{\hd\nu}=
\begin{cases}
0 & \text{branch (i)},\\[2pt]
-\mathcal{E}^{2}\LL_F & \text{branch (ii)},\\[2pt]
-B^{2}\LL_F & \text{branch (iii)}.
\end{cases}
\label{eq:NECbranches}
\end{equation}
The radial electric branch saturates the null energy condition identically, for
every $\LL(F)$ and at every radius; the other two do not. This single observation
already separates branch (i) from the other two, and we take it up in
Sec.~\ref{sec:radial}.

Finally, we record the form taken by the electromagnetic field equations
\eqref{eq:maxwell} in each branch. With $\sqrt{-g}=e^{\Phi}(1-b/r)^{-1/2}r$, the
first of \eqref{eq:maxwell} integrates to
\begin{equation}
r\,E\,\LL_F=q_e \quad\text{(i)},\qquad
e^{\Phi}\LL_F\sqrt{2F}=\text{const}\quad\text{(iii)},
\label{eq:gauss}
\end{equation}
for the radial electric and magnetic branches respectively, $q_e$ being the
electric charge. In the magnetic branch $E=\mathcal{E}=0$, so
\eqref{eq:invariant} gives $F=\tfrac12B^{2}$ and $\sqrt{2F}=|B|$; Eq.~(iii) is
therefore simply $e^{\Phi}\LL_F B=\text{const}$, written through the invariant
rather than through $B$ itself. This is deliberate rather than cosmetic: since
$\LL_F\equiv d\LL/dF$ is by construction a function of $F$, casting the
constraint in terms of $F$ is exactly the form needed in
Sec.~\ref{sec:magnetic} to integrate $\LL_F$ directly into $\LL(F)$ once a
shape function $b(r)$ is chosen, without a separate step converting $B(r)$ into
$F(r)$; writing $\sqrt{2F}$ also sidesteps the sign ambiguity of the
pseudoscalar $B$. Equation~(i) is left with $E$ explicit instead, since no such
reconstruction is needed there: Sec.~\ref{sec:radial} is a non-existence
result for arbitrary $\LL(F)$, and \eqref{eq:gauss}(i) is used only in its
plain Reissner--Nordstr\"om-like form.
In the azimuthal electric branch, by contrast, the first of
\eqref{eq:maxwell} is satisfied identically by staticity and circular symmetry
and imposes no condition at all; the constraint comes instead from the Bianchi
identity, which gives
\begin{equation}
F_{t\varphi}=\text{const}\quad\text{(ii)},\qquad\text{i.e.}\qquad
\mathcal{E}=\frac{\text{const}}{r\,e^{\Phi}} .
\label{eq:bianchi2}
\end{equation}
Ref.~\cite{Cataldo2002} already noted that the first of \eqref{eq:maxwell}, the Gauss-type equation used in \eqref{eq:gauss} for the other two branches, leaves $\mathcal{E}$ completely undetermined in this branch: unlike $E$ and $B$, it satisfies no differential equation of that kind at all, for any $\LL(F)$. What Eq.~\eqref{eq:bianchi2} does is show where the missing
determination of $\mathcal{E}$ actually comes from: not from
the first equation of \eqref{eq:maxwell}, but from its second, the electromagnetic Bianchi identity
$\rd\bm{F}=0$. It reproduces Eq.~(29) of that reference for the linear theory.
The contrast between \eqref{eq:gauss} and \eqref{eq:bianchi2} is the technical
origin of the differences developed in the following sections.

\section{The radial electric branch}\label{sec:radial}

In this branch $E\neq0$ and $\mathcal{E}=B=0$, fixed once and for all by
Proposition~\ref{prop:branches}. We show that no traversable wormhole exists
here, for \emph{any} $\LL(F)$ with $\LL_F\neq0$: the obstruction is purely
geometric and does not depend on the constitutive law of the electromagnetic
field, nor on $\Lambda$.

The starting point is already in hand. Setting $\mathcal{E}=B=0$ in the coupled
system \eqref{eq:coupled-tt}--\eqref{eq:coupled-rr} gives
\begin{gather}
\frac{rb'-b}{2r^{3}}=\Lambda-8\pi\LL-8\pi E^{2}\LL_F ,
\label{eq:radial-restricted-a}\\
\Big(1-\frac{b}{r}\Big)\frac{\Phi'}{r}=8\pi\LL+8\pi E^{2}\LL_F-\Lambda ,
\label{eq:radial-restricted-b}
\end{gather}
and adding these two equations, both $\Lambda$ and the matter content
$8\pi\LL\pm8\pi E^{2}\LL_F$ cancel identically, for every $\LL(F)$ and at every
radius, leaving a statement about the geometry alone,
\begin{equation}
\frac{rb'-b}{2r^{3}}+\Big(1-\frac{b}{r}\Big)\frac{\Phi'}{r}=0 .
\label{eq:radial-ode}
\end{equation}
This cancellation is equivalent to $T_{\hd t\hd t}+T_{\hd r\hd r}=0$, the first
case of \eqref{eq:NECbranches}, but
\eqref{eq:radial-restricted-a}--\eqref{eq:radial-restricted-b} show it
directly at the level of the coupled field equations, without appealing
separately to the sign of $g_{\hd t\hd t}$ and $g_{\hd r\hd r}$.
Equation \eqref{eq:radial-ode} determines $\Phi'$ explicitly in terms of $b(r)$
alone, independently of $\LL(F)$, $E(r)$ and $\Lambda$; it integrates exactly to
\begin{equation}
e^{2\Phi(r)}=C\Big(1-\frac{b(r)}{r}\Big),
\label{eq:radial-solution}
\end{equation}
with $C>0$ an integration constant that can always be set to unity by a
constant rescaling of $t$. Equation \eqref{eq:radial-solution} is the familiar
single-metric-function form shared by the Reissner--Nordstr\"om family and by
its three-dimensional, nonlinear-electrodynamics analogues
\cite{CataldoGarcia,CataldoCruzDelCampoGarcia}; what
\eqref{eq:radial-restricted-a}--\eqref{eq:radial-solution} shows is that in $2+1$
dimensions this form is not a simplifying choice but a theorem, forced by the
branch structure alone.

\begin{theorem}[No traversable wormhole in the radial electric branch]
\label{thm:radial}
Let $\LL(F)$ be any electromagnetic Lagrangian with $\LL_F\neq0$, and let
$\Lambda$ be arbitrary. Every static, circularly symmetric solution of
\eqref{eq:einstein}--\eqref{eq:maxwell} in the radial electric branch that
admits a throat, $b(r_0)=r_0$, satisfies $e^{2\Phi(r_0)}=0$. The throat
coincides with an event horizon, in violation of the no-horizon condition
\eqref{eq:nohorizon}; no traversable wormhole exists in this branch.
\end{theorem}

\begin{proof}
Immediate from \eqref{eq:radial-solution}: at $r=r_0$ the shape function
satisfies $1-b(r_0)/r_0=0$ by definition of the throat, so
$e^{2\Phi(r_0)}=0$ regardless of the value of $C$, of $\Lambda$, and of
the specific form of $\LL(F)$.
\end{proof}

Theorem~\ref{thm:radial} reproduces, for arbitrary $\Lambda$, a conclusion that
Arellano and Lobo already reach directly in $2+1$ dimensions in
Ref.~\cite{ArellanoLobo} for $\Lambda=0$. For transparency we record their relevant equations
here, in their own notation, which coincides with ours on this point. Their
electromagnetic field tensor is taken of the form
\begin{equation}
F_{\mu\nu}=E(r)\big(\delta^t_\mu\delta^r_\nu-\delta^r_\mu\delta^t_\nu\big)
+B(r)\big(\delta^\varphi_\mu\delta^r_\nu-\delta^r_\mu\delta^\varphi_\nu\big),
\label{eq:AL-ansatz}
\end{equation}
their Eq.~(6): note the absence of any $F_{t\varphi}$ term. Their Einstein
tensor components, their Eqs.~(10)--(11), are $G_{\hd t\hd t}=(b'r-b)/2r^{3}$
and $G_{\hd r\hd r}=(1-b/r)\Phi'/r$, in agreement with our
\eqref{eq:Gtt}--\eqref{eq:Grr} once $\Lambda=0$ is set, since their action
carries no cosmological constant. Setting $B(r)=0$, their field equations give
\begin{equation}
\Phi'=-\frac{b'r-b}{2r(r-b)},\qquad\qquad e^{2\Phi}=1-\frac{b}{r},
\label{eq:AL-1920}
\end{equation}
their Eqs.~(19)--(20): precisely the $\Lambda=0$ case of
\eqref{eq:radial-ode}--\eqref{eq:radial-solution}. Their general expression for
the null energy condition combination is
\begin{equation}
T_{\hd\mu\hd\nu}k^{\hd\mu}k^{\hd\nu}
=\frac{1}{8\pi}\left[\frac{b'r-b}{r^{3}}+\Big(1-\frac{b}{r}\Big)\frac{\Phi'}{r}\right],
\label{eq:AL-13}
\end{equation}
their Eq.~(13); comparison with their own Eqs.~(10)--(11) just quoted shows
that the first term should carry a factor of two, $(b'r-b)/2r^{3}$, matching
our \eqref{eq:NECbranches}. The slip does not affect the sign structure at the
throat, which is what their argument, and ours, relies on.

Two things are added here. First, because the $\Lambda$ terms cancel between
$g_{\hd t\hd t}$ and $g_{\hd r\hd r}$, the same conclusion survives verbatim
for arbitrary $\Lambda$, which \eqref{eq:AL-1920} does not address. Second, the
argument is now embedded in the exhaustive classification of
Proposition~\ref{prop:branches}: the ansatz \eqref{eq:AL-ansatz} contains only
$F_{tr}$ and $F_{\varphi r}$, so their non-existence analysis rules out the
coexistence of $E$ and $B$ but does not consider the azimuthal electric branch
at all, the one taken up in Sec.~\ref{sec:electric}. The field equations determine the metric completely, up to \eqref{eq:radial-solution}, before the electromagnetic field equation
 \eqref{eq:gauss}(i) is even invoked, so no
choice of $\LL(F)$ can rescue traversability. The freedom left by
\eqref{eq:maxwell}, the relation $rE\LL_F=q_e$ of \eqref{eq:gauss}, only
fixes how the charge sources $b(r)$ through the remaining field equation; it
plays no role in \eqref{eq:radial-solution} and cannot undo the horizon. In this
sense the branch does not fail to produce new geometry: with $b(r)=r_0$ or
$b(r)=2M$ it reproduces the charged black holes of Ref.~\cite{CataldoGarcia}
and its nonlinear generalisations, which is what the field equations of this
branch generically describe once a horizon rather than a throat is accepted at
$r=r_0$.

With Theorem~\ref{thm:radial} in hand, the radial electric branch is closed:
it saturates the null energy condition, its metric is fixed in the sense of
\eqref{eq:radial-solution}, and it admits no traversable wormhole for any
$\LL(F)$. The remaining two branches, taken up next, evade this obstruction
each for a different reason, since $\mathcal{E}^{2}\LL_F$ and $B^{2}\LL_F$ need
not vanish in \eqref{eq:NECbranches}.

\section{The azimuthal electric branch} \label{sec:electric}

In this branch $E=B=0$ and $\mathcal{E}\neq0$, fixed by
Proposition~\ref{prop:branches}. Unlike the other two branches, the first of
\eqref{eq:maxwell} imposes no condition here, and $\mathcal{E}(r)$ is fixed
instead by the Bianchi identity, \eqref{eq:bianchi2}: writing the integration
constant as $Q$,
\begin{equation}
\mathcal{E}(r)=\frac{Q}{r\,e^{\Phi(r)}} .
\label{eq:azimuthal-field}
\end{equation}

Setting $E=B=0$ in \eqref{eq:coupled-tt}--\eqref{eq:coupled-pp} leaves
$\LL(r)$ and $\mathcal{E}^{2}(r)\LL_F(r)$ to be read off from three equations
for two unknowns. The redundancy is not accidental: adding
\eqref{eq:coupled-tt} and \eqref{eq:coupled-pp}, the matter content cancels
identically between them exactly as in Sec.~\ref{sec:radial}, and so does
$\Lambda$, since $g_{\hd t\hd t}=-1$ and $g_{\hd\varphi\hd\varphi}=+1$ once
more. What remains is a purely geometric constraint valid for arbitrary $\Lambda$:
\begin{equation}
(rb'-b)\,\frac{1-r\Phi'}{2r^{3}}+\Big(1-\frac{b}{r}\Big)\big[\Phi''+(\Phi')^{2}\big]=0 .
\label{eq:azimuthal-master}
\end{equation}

\begin{theorem}[Freedom of the shape function in the azimuthal electric branch]
\label{thm:rigidity}
Let $\Lambda$ be arbitrary. Then the shape function $b(r)$ is unconstrained by
Eq.~\eqref{eq:azimuthal-master} if and only if
\begin{equation}
e^{\Phi(r)}=Cr ,
\label{eq:rigid}
\end{equation}
with $C>0$ a constant that can be set to unity by a constant rescaling of $t$.
\end{theorem} 
\begin{proof}
The coefficient of $b'$ in Eq.~\eqref{eq:azimuthal-master} is
$(1-r\Phi')/(2r^{2})$. Hence, for $\Phi'=1/r$, this coefficient vanishes
identically, and so does the entire term proportional to $(rb'-b)$. The
remaining equation is $(1-b/r)\bigl[\Phi''+(\Phi')^{2}\bigr]=0$,
which is identically satisfied because $\Phi''+(\Phi')^{2}=0$. Therefore,
Eq.~\eqref{eq:azimuthal-master} imposes no restriction on the shape function
$b(r)$.

Conversely, if \eqref{eq:rigid} fails, then $\Phi(r)\neq\ln(Cr)$ for any
constant $C$, so $\Phi'(r)=1/r$ cannot hold at every $r$: there exists some
$r_*$ where $w(r_*)\equiv1-r_*\Phi'(r_*)\neq0$. There the coefficient of
$b'$ does not vanish, so Eq.~\eqref{eq:azimuthal-master} determines
$b'(r_*)$ once $b(r_*)$ is specified, rather than leaving them independent:
$b(r)$ cannot be prescribed freely at $r_*$, and hence is not
unconstrained. This is worked out in full in Sec.~\ref{sec:el-generic}, where solving
\eqref{eq:azimuthal-master} away from \eqref{eq:rigid} gives $b(r)$ fixed
up to a single constant (Eq.~\eqref{eq:b-generic} and
Proposition~\ref{prop:generic-nothroat}).

\end{proof}

With Theorem~\ref{thm:rigidity} fixing $\Phi$, \eqref{eq:coupled-rr} and
\eqref{eq:coupled-pp} reduce to a single relation between $\LL(r)$,
$\LL_F(r)$ and $b(r)$; it can be solved in either direction, and each is
developed in its own subsection below.

\subsection{The Lagrangian from the shape function, $e^{\Phi}=Cr$}\label{sec:el-reconstruction}

With \eqref{eq:rigid} and $C=1$, \eqref{eq:coupled-rr} and
\eqref{eq:coupled-pp} give $\LL(r)$ and $\mathcal{E}^{2}(r)\LL_F(r)$ directly,
using $\mathcal{E}^{2}=Q^{2}/r^{4}$ from \eqref{eq:azimuthal-field}, for
arbitrary $\Lambda$:
\begin{equation}
\LL(r)=\frac{1}{8\pi r^{2}}\Big(1-\frac{b}{r}\Big)+\frac{\Lambda}{8\pi} ,
\label{eq:L-of-r}
\end{equation}
\begin{equation}
\LL_F(r)=\frac{r\big[3b(r)-rb'(r)-2r\big]}{16\pi Q^{2}} .
\label{eq:LF-of-r}
\end{equation}
$\Lambda$ enters only as an additive shift of $\LL(r)$; $\LL_F(r)$, and with
it the whole null-energy structure developed below, is exactly as it would
be with $\Lambda=0$, consistent with the general remark of
Sec.~\ref{sec:setup} that a constant piece of $\LL$ is degenerate with a
shift of $\Lambda$. This is already a first structural fact worth recording: in this direction
of the construction, $b(r)$ is the free input and $\Lambda$ is fully absorbed into $\LL$. Ref.~\cite{MHG}, by contrast, requires $\Lambda<0$ for its fixed $k=1/2$ Lagrangian to admit a static throat at all, a requirement also found for the stationary family of Ref.~\cite{CanateBreton}. In such a way we have shown that in this case a static throat exists for $b(r)$ freely chosen and $\Lambda$ arbitrary, including $\Lambda=0$.

Two consequences follow from the throat condition $b(r_0)=r_0$ alone, for
\emph{any} admissible $b(r)$:
\begin{equation}
\LL(r_0)=\frac{\Lambda}{8\pi} ,\qquad
\LL_F(r_0)=\frac{r_0^{2}\big[1-b'(r_0)\big]}{16\pi Q^{2}}>0 ,
\label{eq:throat-values}
\end{equation}
the second strictly positive by the flare-out condition \eqref{eq:flareout}
and, notably, independent of $\Lambda$. Together with the manifestly finite
$\mathcal{E}(r_0)=Q/r_0^{2}$, the field, $\LL$ and $\LL_F$ are all finite at
the throat; and by the second case of \eqref{eq:NECbranches}, $\LL_F(r_0)>0$
means the null energy condition is violated there, as a throat requires,
again regardless of $\Lambda$.

Because $F=-\mathcal{E}^{2}/2=-Q^{2}/(2r^{4})$ depends on $r$ only through
$\mathcal{E}$, and not on $b(r)$, it inverts independently of the choice of
shape function,
\begin{equation}
r(F)=\Big(\frac{Q^{2}}{-2F}\Big)^{1/4} ,
\label{eq:r-of-F}
\end{equation}
turning \eqref{eq:L-of-r} into a closed-form reconstruction of the Lagrangian
for \emph{any} freely specified $b(r)$ and arbitrary $\Lambda$:
\begin{equation}
\LL(F)=\frac{1}{8\pi\,r(F)^{2}}\left[1-\frac{b\big(r(F)\big)}{r(F)}\right]+\frac{\Lambda}{8\pi} .
\label{eq:reconstruction}
\end{equation}
Since $r(F)\to\infty$ as $F\to0^{-}$, the asymptotic condition
$b(r)=o(r)$, namely
\[
\lim_{r\to\infty}\frac{b(r)}{r}=0,
\]
implies that the bracket converges to unity, so $\LL(F)\to\Lambda/(8\pi)$: with $\Lambda\neq0$ the weak-field
Lagrangian approaches this constant, not zero, and $\Lambda$ is read off
directly from the exterior value of $\LL$. Only for $\Lambda=0$ does the
subleading term become visible,
\begin{equation}
\LL(F)\to\frac{\sqrt{-2F}}{8\pi Q} \qquad (\Lambda=0) ,
\label{eq:weakfield}
\end{equation}
independently of $b(r)$, with the leading weak-field power carrying a
positive coefficient. In either case \eqref{eq:LF-of-r} is unaffected by
$\Lambda$ and turns negative once $b(r)$ falls behind $r$, restoring the null
energy condition away from the throat, for any $b(r)$ growing slower than
$r$.

\emph{Example.} As an application of the procedure just described, take the
simplest admissible shape function, $b(r)=r_0$ constant. This is the same shape function taken up again in Sec.~\ref{sec:el-generic}, where it is shown to admit other redshift functions besides $\Phi=\ln(Cr)$. Conditions \eqref{eq:domain} and \eqref{eq:flareout} hold trivially
throughout $r>r_0$, for any $\Lambda$, and Eqs.~\eqref{eq:L-of-r}--\eqref{eq:LF-of-r}
give
\begin{equation}
\LL(r)=\frac{r-r_0}{8\pi r^{3}}+\frac{\Lambda}{8\pi} ,\qquad
\LL_F(r)=\frac{r(3r_0-2r)}{16\pi Q^{2}}. \label{Eq.45}
\end{equation}

Equation \eqref{eq:reconstruction} then reconstructs the Lagrangian in
closed form,
\begin{equation}
\LL(F)=\underbrace{\frac{\sqrt{-2F}}{8\pi Q}}_{\text{universal}}
-\underbrace{\frac{r_0(-2F)^{3/4}}{8\pi Q^{3/2}}}_{\text{sourced by the throat}}
+\underbrace{\frac{\Lambda}{8\pi}}_{\text{constant}} .
\label{eq:L-explicit}
\end{equation}
The first term of \eqref{eq:L-explicit} is universal: it comes entirely
from the first term of Eq. \eqref{eq:reconstruction}, which carries no dependence on $b(r)$
at all, so every reconstruction in this branch produces exactly
$\sqrt{-2F}/(8\pi Q)$, regardless of which shape function was chosen. The
second term comes from the $-b(r)/r$ part, and its power depends on $b(r)$
itself; for the constant $b(r)=r_0$ used here it is $(-2F)^{3/4}$, but a
differently chosen $b(r)$ would produce a different power in its place.
This second term is what carries $r_0$, and vanishes if $r_0=0$: it is
generated specifically by this choice of throat; the third is simply $\Lambda$, as already
seen in \eqref{eq:L-of-r}. Since its coefficients contain the charge $Q$ of the solution used as input, \eqref{eq:L-explicit} is the Lagrangian reconstructed for this particular solution rather than a theory in its own right; in Sec.~\ref{sec:el-theory} its functional form is adopted as the definition of a theory with fixed coupling constants, of which the present solution is one member.

Having reconstructed $\LL(F)$, the quantity that actually controls the null
energy condition is its derivative, $\LL_F(r)$, already positive at the
throat by \eqref{eq:throat-values}. Whether that violation persists away from the throat depends on the sign of
$\LL_F(r)$ itself. From the second expression of \eqref{Eq.45}, we see that $\LL_F(r)$ is positive for $r_0\le r<\tfrac32r_0$ and negative for $r>\tfrac32r_0$,
crossing zero once at
\begin{equation}
r_1=\tfrac32 r_0 .
\end{equation}
The exotic matter needed to hold the throat open is therefore confined to
the finite shell $r_0\le r<r_1$; beyond $r_1$ the same field satisfies the
null energy condition, as ordinary matter does. Since $\LL_F$ carries no
dependence on $\Lambda$, this shell boundary is the same for every
$\Lambda$.

The first term of \eqref{eq:L-explicit} contributes nothing to
$\LL-2F\LL_F$, which is precisely the combination
$\LL+\mathcal{E}^{2}\LL_F$ appearing in \eqref{eq:coupled-tt}, the field
equation coupling $b(r)$ and $\LL(F)$. From
\eqref{eq:L-of-r}--\eqref{eq:LF-of-r},
\begin{equation}
\LL+\mathcal{E}^{2}\LL_F = \frac{\Lambda}{8\pi}-\frac{rb'-b}{16\pi r^{3}} .
\end{equation}
The universal piece of \eqref{eq:L-explicit} is therefore inert as
far as the shape function is concerned: it drops out of this combination
completely, whatever its coefficient, and plays no role in sourcing the
throat. All of the work is done by the remaining, non-degenerate part of the
reconstructed Lagrangian: here, the $(-2F)^{3/4}$ term generated by the
throat itself. Section~\ref{sec:el-forward} examines the same combination from the opposite
direction, fixing $\LL(F)$ first and solving for $b(r)$.

\subsection{The shape function from the Lagrangian, $e^{\Phi}=Cr$}\label{sec:el-forward}

The relation between \eqref{eq:L-of-r} and \eqref{eq:reconstruction} is
equally well read the other way. In Sec.~\ref{sec:el-reconstruction}, the shape function $b(r)$
was the free input and $\LL(F)$ the result; here the two roles simply
reverse, and $\LL(F)$ may be chosen first, for any $\Lambda$, with $b(r)$
solved for directly. This reversal is only meaningful because $\LL(F)$ is
left genuinely unrestricted, free to be any function with $\LL_F\neq0$:
choosing it is now the free step, mirroring exactly the freedom $b(r)$
had in Sec.~\ref{sec:el-reconstruction}.

Inverting \eqref{eq:L-of-r},
\begin{equation}
b(r)=r-8\pi r^{3}\LL\big(F(r)\big)+\Lambda r^{3} ,
\label{eq:b-of-LL}
\end{equation}
with $F(r)=-Q^{2}/(2r^{4})$ from \eqref{eq:azimuthal-field} as before,
independently of $b(r)$: any $\LL(F)$ with $\LL_F\neq0$ and any $\Lambda$,
substituted into \eqref{eq:b-of-LL}, generates a solution of the full coupled
system through Theorem~\ref{thm:rigidity} alone. Here $\Lambda$ no longer
drops out: for a fixed, given $\LL(F)$, it genuinely reshapes $b(r)$, as the
example below shows.

Equation \eqref{eq:b-of-LL} makes the location of the throat a joint property
of the chosen Lagrangian and $\Lambda$: $b(r_0)=r_0$ holds if and only if
\begin{equation}
\LL\big(F(r_0)\big)=\frac{\Lambda}{8\pi}.
\label{eq:throat-from-LL}
\end{equation}
Thus, a throat exists whenever the chosen Lagrangian crosses the level $\Lambda/(8\pi)$, which coincides with the zero of $\LL$ only for $\Lambda=0$. The throat is then located at the radius $r_0$ corresponding to that crossing through $F(r)=-Q^{2}/(2r^{4})$. Differentiating \eqref{eq:b-of-LL} and using \eqref{eq:throat-from-LL} to
eliminate $\LL(F(r_0))$, the explicit $\Lambda r^{3}$ term and the one it
removes cancel exactly, leaving
\begin{equation}
b'(r_0)=1-\frac{16\pi Q^{2}}{r_0^{2}}\,\LL_F\big(F(r_0)\big) ,
\end{equation}
with no leftover dependence on $\Lambda$ at all: setting $\Lambda=0$ from
the outset would have produced this same expression term by term, since
$\LL_F$ never carried any $\Lambda$-dependence to begin with, as already
noted in Sec.~\ref{sec:el-reconstruction}. The flare-out condition \eqref{eq:flareout} therefore holds if and only if
$\LL_F(F(r_0))>0$, regardless of $\Lambda$: the Lagrangian must be
\emph{increasing} as it crosses the level \eqref{eq:throat-from-LL}, not
just crossing it. By \eqref{eq:NECbranches}, that same sign,
$\LL_F(F(r_0))>0$, is exactly what makes the null energy condition
violated at $r_0$. A throat and a genuine violation of the null energy condition at that
throat remain the same requirement on $\LL(F)$, regardless of $\Lambda$. The
domain condition \eqref{eq:domain} for $r>r_0$ becomes, through
\eqref{eq:b-of-LL}, $\LL(F(r))>\Lambda/(8\pi)$ throughout $F\in(F(r_0),0)$:
the Lagrangian must stay above that level all the way out to $F=0$.

\emph{Example: the single power law.} In Sec.~\ref{sec:el-reconstruction}, $\Lambda$ never affected
the geometry: $b(r)$ was the free input, and $\Lambda$ only shifted $\LL(r)$
additively. To see that this is a feature of that direction of the
construction, not a general fact, fix the Lagrangian once and for all,
\begin{equation}
\LL(F)=-\alpha(-2F)^{k} ,
\label{eq:power-law}
\end{equation}
the ansatz used throughout the literature on this branch, with $\alpha,k>0$
held fixed from here on. Here $\LL_F(F)=2\alpha k(-2F)^{k-1}$, positive for
every $F<0$, so flare-out holds automatically wherever a throat exists, by
the general criterion above. With $F(r)=-Q^{2}/(2r^{4})$ from
\eqref{eq:azimuthal-field}, \eqref{eq:b-of-LL} gives
\begin{equation}
b(r) = r + \Lambda r^{3} + 8\pi\alpha Q^{2k}\,r^{3-4k} ,
\label{eq:general-k-b}
\end{equation}
with only two free constants left, $\Lambda$ and $\alpha$: unlike
Sec.~\ref{sec:el-reconstruction}, where $b(r)$ could be any function, fixing $k$ fixes the entire
functional form of $b(r)$, for any $k$, leaving nothing but these two
constants to adjust. Here $\Lambda$ no longer drops out, as it did in
Sec.~\ref{sec:el-reconstruction}: it appears explicitly in \eqref{eq:general-k-b} and genuinely
reshapes the geometry. The throat condition \eqref{eq:throat-from-LL} fixes
\begin{equation}
\Lambda = -8\pi\alpha\Big(\frac{Q}{r_0^{2}}\Big)^{2k} ,
\label{eq:general-k-Lambda}
\end{equation}
strictly negative for every $k>0$: for this fixed theory, choosing
$\Lambda<0$ locates the throat through \eqref{eq:general-k-Lambda}, while
$\Lambda=0$ admits no throat at all. Note also that, since $\LL_F(F)>0$ for
every $F<0$, the null energy condition is violated throughout the exterior
rather than confined to a finite shell as in the example of Sec.~\ref{sec:el-reconstruction}; confining
it would require $\LL(F)$ to stop increasing before $F=0$, which no single
power does. Shell confinement is therefore not a generic feature of this
direction of the construction.

\emph{The case $k=1/2$.} Taking $k=1/2$ in \eqref{eq:power-law} recovers the
Lagrangian used, up to normalisation, by Refs.~\cite{CanateBreton,MHG}.
Equation \eqref{eq:general-k-b} then reads
\begin{equation}
b(r)=\mu\,r+\Lambda r^{3} ,\qquad \mu\equiv1+8\pi\alpha Q ,
\label{eq:khalf-b}
\end{equation}
and \eqref{eq:general-k-Lambda} becomes $\Lambda=-8\pi\alpha Q/r_0^{2}$.
Note that the field's own contribution to $b(r)$ has the same radial
dependence as the term already present at $\Lambda=0$, and the two merge
into the single term $\mu r$ instead of remaining distinct. For any other
$k$, the third term of \eqref{eq:general-k-b} keeps its own power of $r$, so
$b(r)$ has three genuinely different radial dependences rather than two;
\eqref{eq:general-k-Lambda} still forces $\Lambda<0$ regardless. For
$k=1/2$ specifically, $\mu>1$ identically, for any $\alpha,Q>0$, which is
the merger's own signature: the throat condition gives
$r_0^{2}=(1-\mu)/\Lambda$, real and positive only for $\Lambda<0$, and at
$\Lambda=0$ \eqref{eq:khalf-b} reduces to $b(r)=\mu r$ with $\mu\neq1$,
which has no throat at finite $r$ at all. This is exactly the restriction
reported in Ref.~\cite{MHG}, now seen as one instance of a
general obstruction rather than a pathology special to $k=1/2$, at least in the static sector: fixing
$\LL(F)$ to any single power, whichever it is, leaves $b(r)$ determined and
forces $\Lambda\neq0$; escaping this requires abandoning the restriction to
a single power altogether, which is exactly what letting $\LL(F)$ vary
freely, as in this subsection, achieves.

Equation \eqref{eq:khalf-b} is also recognisable: writing $M\equiv\mu-1=8\pi\alpha Q$,
it reads $b(r)=r(1+M+\Lambda r^{2})$, precisely the mass function of the
static BTZ black hole. This coincidence is worth spelling out carefully,
since $b(r)$ itself does not decide which redshift function accompanies
it; that is fixed independently, by whatever sources the geometry.

In vacuum, $\LL\equiv0$, so $T_{\hd\mu\hd\nu}\equiv0$ for any field
configuration, and the same $G_{\hd t\hd t}+G_{\hd r\hd r}=0$ mechanism
behind \eqref{eq:radial-solution} applies regardless of $b(r)$: it forces
$e^{2\Phi}=1-b/r$, not \eqref{eq:rigid}. Only once this is settled does
substituting $b(r)=r(1+M+\Lambda r^{2})$ show what it gives: a horizon at
$r_0$ for $M>0$, the BTZ black hole itself.

In the azimuthal branch with a genuine field, $\LL_F\neq0$,
Theorem~\ref{thm:rigidity} instead forces $e^{\Phi}=Cr$, again
independently of $b(r)$ and before any Lagrangian or shape function is
chosen. Only afterwards, solving \eqref{eq:coupled-rr} for the specific
choice $\LL(F)=-\alpha\sqrt{-2F}$, does $b(r)$ turn out to be this same
function \eqref{eq:khalf-b}.

The comparison is therefore between two independent constructions that
happen to produce the same $b(r)$: vacuum forces $e^{2\Phi}=1-b/r$ and gives
a horizon at $r_0$; the nonlinear electric field forces $e^{\Phi}=Cr$
instead and gives a genuine, traversable throat at the same $r_0$. The
field does not move where the vacuum horizon would have sat, it removes it
by replacing the redshift function that would have produced it.

Traversable wormholes with an azimuthal nonlinear electric field have so
far been reported only for the power-Maxwell Lagrangian $\LL\propto|F|^{k}$
with the single exponent $k=1/2$: Refs.~\cite{CanateBreton,MHG}, the first stationary and the second static. In our work the exponent is kept arbitrary, and the static solution \eqref{eq:general-k-b} is obtained for any $k$; the geometry of Ref.~\cite{MHG} is recovered as the particular case $k=1/2$, while the stationary family of Ref.~\cite{CanateBreton}, having nonzero angular momentum, lies outside the static sector and is not recovered here. In the static case $k=1/2$ the field's contribution to $b(r)$ degenerates into the term already present at $\Lambda=0$, the same degeneracy $\LL-2F\LL_F=0$ noted in the Introduction, here showing up as the field-sourced term $r^{3-4k}$ collapsing onto the vacuum term $r$.

\subsection{A theory from the reconstructed Lagrangian}\label{sec:el-theory}
The coefficients of \eqref{eq:L-explicit} contain the charge $Q$ of the
solution used as input to reconstruct the Lagrangian of the azimuthal
electric branch, so that \eqref{eq:L-explicit} is the Lagrangian
reconstructed for that particular solution rather than a theory in its own
right. The physics, however, does not depend on $Q$. For a given $b(r)$,
$\LL(r)$ in \eqref{eq:L-of-r} is independent of $Q$, and although
$\LL_F(r)$ in \eqref{eq:LF-of-r} carries a factor $1/Q^{2}$, it enters the
stress tensor only through the combination $\mathcal{E}^{2}\LL_F$, in which
this factor cancels against $\mathcal{E}^{2}=Q^{2}/r^{4}$,
\begin{equation}
\mathcal{E}^{2}\LL_F=\frac{3b(r)-rb'(r)-2r}{16\pi r^{3}} .
\end{equation}
The geometry and its source are therefore the same for every value of $Q$,
which only sets, through \eqref{eq:r-of-F}, the scale on which $F$ is
measured, and for this reason appears in the coefficients of the
reconstructed $\LL(F)$. 

What the reconstruction \eqref{eq:L-explicit} determines is the functional form of $\LL(F)$, and we adopt that form as the definition of the theory,
\begin{equation}
\LL(F)=\alpha_1\sqrt{-2F}-\alpha_2\,(-2F)^{3/4} ,\qquad \alpha_1,\alpha_2>0 ,
\label{eq:two-power}
\end{equation}
where $\alpha_1$ and $\alpha_2$ are fixed coupling constants, independent of
the charge of any solution, as $\beta$ is in Born-Infeld electrodynamics.
Since $\LL(0)=0$, the nonlinear contribution vanishes together with the
field, and only Einstein gravity remains when the field is switched off. The
constant $\Lambda/(8\pi)$ of \eqref{eq:L-explicit} has been dropped, since it
only cancels the cosmological term of \eqref{eq:action}.

We now solve the theory \eqref{eq:two-power} in the direction of
Sec.~\ref{sec:el-forward}, with the Lagrangian given and the shape function
unknown. Theorem~\ref{thm:rigidity} gives $e^{\Phi}=r$, and the charge $Q$
now enters only through $\mathcal{E}=Q/r^{2}$, as a free integration
constant of the solution, while $\alpha_1$ and $\alpha_2$ stay fixed.
Equation~\eqref{eq:b-of-LL}, with $F(r)=-Q^{2}/(2r^{4})$, then gives
\begin{equation}
b(r)=\big(1-8\pi\alpha_1 Q\big)\,r+8\pi\alpha_2\,Q^{3/2}+\Lambda r^{3} .
\label{eq:two-power-b}
\end{equation}
This is a family of solutions of one and the same theory, labelled by the
charge $Q$ and the cosmological constant $\Lambda$.

Although the Einstein equations are nonlinear, the dependence of $b(r)$ on
the Lagrangian in \eqref{eq:two-power-b} is not. In the Morris--Thorne form
\eqref{eq:metric} the nonlinearity of the Einstein tensor resides entirely
in the redshift function. For a given $\Phi(r)$, the components
$G^{a}{}_{b}$ are linear in $b(r)$. In the present regime, $e^{\Phi}=r$ is
fixed by Theorem~\ref{thm:rigidity} and $\mathcal{E}=Q/r^{2}$ by the Bianchi
identity, both independently of the Lagrangian, so that
$F(r)=-Q^{2}/(2r^{4})$ is fixed as well, and the stress tensor, with
$T^{r}{}_{r}=\LL$ and $T^{t}{}_{t}=T^{\varphi}{}_{\varphi}=\LL+\mathcal{E}^{2}\LL_F$, is linear in $\LL$. The $rr$ equation then reduces to the algebraic relation \eqref{eq:b-of-LL}, and the remaining equations follow from it. Because $e^{\Phi}=r$ and $\mathcal{E}=Q/r^{2}$ are fixed independently of
the Lagrangian, Eq.~\eqref{eq:b-of-LL} is affine in $\LL$: for
$\LL=\LL_1+\LL_2$ it gives
$b(r)=r+\Lambda r^{3}-8\pi r^{3}\LL_1\big(F(r)\big)-8\pi r^{3}\LL_2\big(F(r)\big)$,
so that the vacuum part $r+\Lambda r^{3}$ appears once and the field
contributions of the two terms simply add. Equation \eqref{eq:two-power-b}
is thus the sum of the $k=1/2$ and $k=3/4$ contributions of
\eqref{eq:general-k-b}, with coefficients of opposite sign
($\alpha=-\alpha_1$ and $\alpha=\alpha_2$). Unlike either single power, this
combination admits throats with $\Lambda=0$ and confines the exotic matter
to a finite shell. The additivity holds only in the regime $e^{\Phi}=r$; in
the generic regime of Sec.~\ref{sec:el-generic}, where $\Phi$ is not fixed in
advance, it is lost.

The example of Sec.~\ref{sec:el-reconstruction} is contained in this family
as a particular case. Take $\Lambda=0$ and the member whose charge satisfies
$8\pi\alpha_1Q=1$. The linear term of \eqref{eq:two-power-b} then vanishes,
and the shape function reduces to the constant
\begin{equation}
b(r)=8\pi\alpha_2\,Q^{3/2}\equiv r_0 .
\label{eq:two-power-b-example}
\end{equation}
This is exactly the shape function chosen as input in
Sec.~\ref{sec:el-reconstruction}, now obtained as the output of the field
equations of a fixed theory, with the same redshift function $e^{\Phi}=r$ and
the same field $\mathcal{E}=Q/r^{2}$.

Read the other way, the two conditions
$8\pi\alpha_1Q=1$ and $8\pi\alpha_2Q^{3/2}=r_0$ give
\begin{equation}
\alpha_1=\frac{1}{8\pi Q} ,\qquad \alpha_2=\frac{r_0}{8\pi Q^{3/2}} ,
\label{eq:alphas-example}
\end{equation}
which are precisely the coefficients of \eqref{eq:L-explicit}. Accordingly,
$\LL(r)$ and $\LL_F(r)$ along this solution reproduce \eqref{Eq.45} with
$\Lambda=0$, and the exotic matter is confined to the same shell
$r_0\le r<\tfrac32r_0$. The remaining members of the family are further
solutions of the same theory, with the same $\alpha_1$ and $\alpha_2$ but a
different charge. For $\Lambda=0$, \eqref{eq:two-power-b} gives
\begin{equation}
1-\frac{b(r)}{r}=8\pi\alpha_1Q\,\Big(1-\frac{r_{\rm th}}{r}\Big),
\end{equation}
a traversable wormhole with throat at $r_{\rm th}=\frac{\alpha_2}{\alpha_1}\,Q^{1/2}$ for every $Q>0$, which is
the example of Sec.~\ref{sec:el-reconstruction} when $8\pi\alpha_1Q=1$, so
that $r_{\rm th}=r_0$. The wormhole of Sec.~\ref{sec:el-reconstruction} is
therefore not an isolated configuration with a Lagrangian tailored to it,
but one member of a family of solutions of a single theory, in which the
charge can be varied freely without changing the Lagrangian.

More generally, \eqref{eq:two-power-b} describes a traversable wormhole for
every $Q>0$ and $\Lambda\le0$, and not only for the member singled out
above. Writing $a\equiv8\pi\alpha_1Q$ and $c\equiv8\pi\alpha_2Q^{3/2}$, the
function $b(r)-r=c-ar+\Lambda r^{3}$ is positive at $r=0$ and strictly
decreasing, so there is a single throat $r_0$, with $r_0=c/a$ for
$\Lambda=0$ and $r_0<c/a$ for $\Lambda<0$; there $b'(r_0)=1-a+3\Lambda
r_0^{2}<1$, and $1-b/r>0$ for all $r>r_0$. Since $e^{\Phi}=r$ never
vanishes, there is no horizon, and the field $\mathcal{E}(r_0)=Q/r_0^{2}$ is
finite at the throat. Along the solution
$\mathcal{E}^{2}\LL_F=(3\alpha_2Q^{3/2}-2\alpha_1Q\,r)/(2r^{3})$, which is
positive at $r_0\le c/a$ and negative for $r>3\alpha_2Q^{1/2}/(2\alpha_1)$: the
null energy condition is violated at the throat and the exotic matter is
again confined to a finite shell. Asymptotically, $1-b/r\simeq|\Lambda|r^{2}$
for $\Lambda<0$, while for $\Lambda=0$ one has $1-b/r\to a$, so that for
$a\neq1$ the spatial section approaches a cone,
$\rd\ell^{2}\simeq\rd r^{2}/a+r^{2}\rd\varphi^{2}$, with an angular deficit
($a<1$) or excess ($a>1$); the example of Sec.~\ref{sec:el-reconstruction}
is the case $a=1$. For $\Lambda>0$, instead, $1-b/r$ becomes negative at
large $r$, and the region outside the throat does not extend to infinity.
This family thus provides, within a single nonlinear theory, static
traversable wormholes with an azimuthal electric field for both
$\Lambda=0$ and $\Lambda<0$.

\subsection{The generic branch, $\Phi'\neq1/r$}\label{sec:el-generic}
The constructions of Secs.~\ref{sec:el-reconstruction} and \ref{sec:el-forward} rest entirely on \eqref{eq:rigid}.
Before moving to the magnetic branch, it is worth asking how essential that
choice really is: does every other redshift function fail to support a
throat, or did Theorem~\ref{thm:rigidity} merely pick out the most
convenient one? Away from \eqref{eq:rigid}, \eqref{eq:azimuthal-master} is a
genuine first-order linear ODE for $b(r)$, for any given $\Phi(r)$; its
general solution is
\begin{equation}
b(r) = \frac{c_1\,r\,e^{-2\Phi(r)}}{w(r)^{2}} + r , \qquad
w(r)\equiv1-r\Phi'(r) ,
\label{eq:b-generic}
\end{equation}
with $c_1$ a single integration constant. Unlike the case of Theorem~\ref{thm:rigidity}, $b(r)$ is here fixed
up to one constant, not a free function; \eqref{eq:b-generic} degenerates
back into the freedom of \eqref{eq:rigid} exactly where $w(r)=0$.

With \eqref{eq:b-generic}, Eqs.~\eqref{eq:coupled-tt} and \eqref{eq:coupled-rr} give, for arbitrary $\Lambda$,
\begin{equation}
\LL(r) = \frac{\Lambda}{8\pi} -\frac{c_1\,\Phi'(r)\,e^{-2\Phi(r)}}{8\pi r\,w(r)^{2}} ,
\label{eq:L-generic}
\end{equation}
\begin{equation}
\LL_F(r) = -\frac{c_1\,r\big[2r(\Phi')^{2}+r\Phi''-\Phi'\big]}{8\pi
Q^{2}\,w(r)^{3}} .
\label{eq:LF-generic}
\end{equation}
As in Sec.~\ref{sec:el-reconstruction}, $\Lambda$ enters only as an additive
shift of $\LL(r)$; $\LL_F(r)$ carries no dependence on $\Lambda$ at all,
since the $\Lambda$-terms cancel identically between
\eqref{eq:coupled-tt} and \eqref{eq:coupled-rr} when solving for it. Both
terms of \eqref{eq:LF-generic} are proportional to $c_1$; at $c_1=0$,
$b(r)\equiv r$, so $1-b/r\equiv0$ and the metric is singular everywhere,
not a spacetime at all. $\LL_F(r)$ carries no explicit dependence on
$\Phi$ beyond $\Phi'$ and $\Phi''$.

\begin{proposition}[No throat in the generic branch]
\label{prop:generic-nothroat}
If $c_1\neq0$ and $w(r)\neq0$ throughout the domain, then $b(r)\neq r$ for
every $r$: no throat exists.
\end{proposition}

\begin{proof}
Equation~\eqref{eq:b-generic} implies $1-\frac{b(r)}{r}=-\frac{c_1e^{-2\Phi(r)}}{w(r)^2}$.
Since $e^{-2\Phi(r)}$ and $w(r)^2$ are strictly positive wherever
$w(r)\neq0$, the right-hand side is sign-definite, with sign determined solely
by $-c_1$. It follows that $1-b(r)/r$ cannot vanish, and therefore
$b(r)\neq r$ for all $r$.
\end{proof}

A throat can therefore only appear at a point $r_0$ where $w(r_0)$ diverges, that is, where $\Phi'(r_0)\to\infty$. Equation \eqref{eq:b-generic} was derived assuming $w$ finite, so it no longer applies there; such a point must instead be examined directly in \eqref{eq:azimuthal-master}.

\emph{The case $b(r)=r_0$ constant.} Take $b(r)=r_0$, constant, the
simplest nontrivial choice: with $b'=0$, \eqref{eq:azimuthal-master} reduces
to a second-order ODE for $\Phi(r)$ alone. Its solution is
\begin{equation}
e^{\Phi(r)} = C_1\,r + C_2\sqrt{r(r-r_0)} , \qquad r\geq r_0 ,
\label{eq:Phi-const-b}
\end{equation}
with $C_1,C_2$ integration constants. At $C_2=0$ this is exactly the solution $e^{\Phi}=C_1r$ of Sec.~\ref{sec:el-reconstruction}; for $C_2\neq0$ it is a genuinely different redshift function, still sharing the same constant shape function. Note that $C_1$ cannot be set to zero: doing so makes $e^{2\Phi(r)}$ develop an event horizon at $r_0$ rather than a throat.

Take $C_1>0$. At the throat, $e^{\Phi(r_0)}=C_1r_0$ is finite and positive,
and if also $C_2\geq0$ then $e^{\Phi(r)}\geq C_1r\geq C_1r_0>0$ throughout
$r\geq r_0$, so no horizon appears anywhere.

At first sight the throat looks singular: $\Phi'(r_0)$ and $\Phi''(r_0)$
both diverge, since $\sqrt{r-r_0}$ is not a smooth function of $r$ there.
This is not a curvature singularity, only the coordinate $r$ behaving badly,
exactly as $g_{\hd r\hd r}=1/(1-b/r)$ always does at a Morris--Thorne
throat. Switching to the proper radial distance $\ell$, defined by
$\rd\ell/\rd r=1/\sqrt{1-b/r}$, removes the problem: direct computation gives
\begin{equation}
\left.\frac{\rd\Phi}{\rd\ell}\right|_{r_0} = \frac{C_2}{2C_1r_0} ,
\label{eq:dPhi-dell}
\end{equation}
finite, and $\rd^2\Phi/\rd\ell^2|_{r_0}$ is finite as well. The curvature
components confirm this directly: $G_{\hd r\hd r}(r_0)=0$ and
$G_{\hd\varphi\hd\varphi}(r_0)=1/(2r_0^{2})$ from \eqref{eq:Grr} and
\eqref{eq:Gpp}, both finite. The redshift is singular only in the bad
coordinate $r$, never physically.

The matter content follows from \eqref{eq:coupled-rr}--\eqref{eq:coupled-pp}
directly, without needing \eqref{eq:L-generic}--\eqref{eq:LF-generic} (which
assumed $w$ finite and so does not apply at the throat itself):
\begin{equation}
\LL(r_0)=\frac{\Lambda}{8\pi} , \qquad
\LL_F(r_0)=\frac{C_1^{2}r_0^{2}}{16\pi Q^{2}}>0 .
\label{eq:const-b-throat-values}
\end{equation}
As at the throat of the $e^{\Phi}=Cr$ branch, $L(r_0)=\Lambda/(8\pi)$: the null energy
condition is violated at the throat, as required, and every quantity
checked is finite. So \eqref{eq:Phi-const-b} describes a genuine traversable
wormhole for every $C_2\geq0$, with the same shape function as the example of Sec.~\ref{sec:el-reconstruction} but, for $C_2\neq0$, a different redshift function.

What does not carry over is the closed-form reconstruction of $\LL(F)$.
Because $e^{2\Phi(r)}=\big(C_1r+C_2\sqrt{r(r-r_0)}\big)^{2}$ is no longer
proportional to $r^{2}$, the invariant
\begin{equation}
F(r)=-\frac{Q^{2}}{2r^{2}\big(C_1r+C_2\sqrt{r(r-r_0)}\big)^{2}}
\end{equation}
is not a pure power of $r$, and for $C_2\neq0$ inverting it to get $r(F)$
has no elementary closed form. The reconstruction \eqref{eq:reconstruction}
is a special feature of the $e^{\Phi}=Cr$ branch, where $F\propto r^{-4}$ holds
exactly, not a generic property of every traversable solution in this
branch.

\section{The magnetic branch}\label{sec:magnetic}

In this branch $E=\mathcal{E}=0$ and $B\neq0$, fixed by
Proposition~\ref{prop:branches}.
Setting $E=\mathcal{E}=0$ in \eqref{eq:coupled-rr} and \eqref{eq:coupled-pp}
leaves both with the \emph{same} right-hand side,
$8\pi\LL-8\pi B^{2}\LL_F-\Lambda$; unlike the electric branches, it is this
pair, not \eqref{eq:coupled-tt}, that coincides here. Subtracting them and
dividing by $(1-b/r)$ gives a relation between $\Phi$ and $b(r)$ alone,
matter- and $\Lambda$-independent:
\begin{equation}
2r(r-b)\big[\Phi''+(\Phi')^{2}\big] = \Phi'\big(rb'+2r-3b\big) .
\label{eq:magnetic-master}
\end{equation}

\begin{theorem}[Freedom of the shape function in the magnetic branch]
\label{thm:magnetic-rigidity}
Let $\Lambda$ be arbitrary. Then the shape function $b(r)$ is unconstrained
by Eq.~\eqref{eq:magnetic-master} if and only if
\begin{equation}
\Phi(r) = \Phi_0 ,
\label{eq:magnetic-rigid}
\end{equation}
a constant, which can be set to zero by a constant rescaling of $t$.
\end{theorem}

\begin{proof}
Expanding \eqref{eq:magnetic-master} as a polynomial in $b$ and $b'$ at
fixed $r$, the coefficient of $b'$ is $-r\Phi'$. For this to vanish
identically, $\Phi'=0$; the coefficient of $b$ is then $-2r\Phi''$, forcing
$\Phi''=0$ as well, and with both derivatives zero the remaining constant
term vanishes automatically. Conversely, if $\Phi'\not\equiv0$, the
coefficient of $b'$ is nonzero at some $r_*$, so \eqref{eq:magnetic-master}
determines $b'(r_*)$ once $b(r_*)$ is specified, and $b(r)$ is not
unconstrained.
\end{proof}

With Theorem~\ref{thm:magnetic-rigidity} fixing $\Phi\equiv0$,
\eqref{eq:coupled-tt} and \eqref{eq:coupled-rr} reduce, as in
Sec.~\ref{sec:electric}, to a single relation between $\LL(r)$, $\LL_F(r)$
and $b(r)$. Here, however, the two directions are not symmetric: the
Gauss-type relation \eqref{eq:gauss}(iii) ties $B(r)$ to $\LL_F$ directly,
with no analogue of the azimuthal branch's $\mathcal{E}(r)=Q/(re^{\Phi})$
fixed by Bianchi alone. Sections~\ref{sec:mag-reconstruction} and
\ref{sec:mag-power} develop each direction in turn, and the asymmetry
between them.

Note that this result confines the traversable wormholes admitted in this branch to the zero-tidal-force class of Morris and Thorne, since a constant $\Phi$ eliminates precisely the tidal-force terms in the geodesic equation for a traveler crossing the wormhole.

\subsection{The Lagrangian from the shape function}\label{sec:mag-reconstruction}

With $\Phi\equiv0$, \eqref{eq:coupled-tt} gives $\LL(r)$ directly, and
combining it with \eqref{eq:coupled-rr} isolates $B^{2}\LL_F$; using
\eqref{eq:gauss}(iii), $\LL_F B=Q_m$ for an integration constant $Q_m$,
these close into
\begin{equation}
B(r) = \frac{b(r)-rb'(r)}{16\pi Q_m r^{3}} ,
\label{eq:magnetic-B}
\end{equation}
\begin{equation}
\LL(r) = \frac{\Lambda}{8\pi} + Q_m B(r) , \qquad
\LL_F(r) = \frac{Q_m}{B(r)} ,
\label{eq:magnetic-LLF}
\end{equation}
for \emph{any} freely chosen $b(r)$ and arbitrary $\Lambda$: unlike the
electric branches, $B(r)$ is not fixed independently of $b(r)$ by Bianchi
alone, so choosing $b(r)$ fixes the field together with the Lagrangian in
one step. The no-horizon condition \eqref{eq:nohorizon} is automatic here,
since $\Phi\equiv0$ is finite everywhere; only the domain and flare-out
conditions on $b(r)$ remain to be checked, exactly as in Sec.~\ref{sec:setup}.

Because $F=B^{2}/2$, \eqref{eq:magnetic-LLF} reads
$\LL(r)=\Lambda/(8\pi)+Q_m\sqrt{2F(r)}$ for every $b(r)$ satisfying
flare-out, $b(r)>rb'(r)$, which makes $B(r)>0$ and the square root
unambiguous. This is already the reconstruction: unlike
Sec.~\ref{sec:el-reconstruction}, where different shape functions produced
genuinely different combinations of powers, every admissible $b(r)$ in this
branch reconstructs the \emph{same} family of nonlinear electrodynamics,
\begin{equation}
\LL(F) = \frac{\Lambda}{8\pi} + Q_m\sqrt{2F} ,
\label{eq:magnetic-reconstruction}
\end{equation}
with the freedom in $b(r)$ carried entirely by $B(r)$ inside $F(r)$, not by
the functional form of $\LL$. What is new from one $b(r)$ to the next is
therefore the geometry, not the theory sourcing it.

\emph{Example.} Take $b(r)=r_0$ constant. Equation \eqref{eq:magnetic-B}
gives $B(r)=r_0/(16\pi Q_m r^{3})$, so \eqref{eq:magnetic-LLF} and
\eqref{eq:magnetic-reconstruction} become
\begin{equation}
\LL(r)=\frac{\Lambda}{8\pi}+\frac{r_0}{16\pi r^{3}} , \qquad
\LL_F(r)=\frac{16\pi Q_m^{2}r^{3}}{r_0} ,
\end{equation}
matching \eqref{eq:magnetic-reconstruction} with $r=(r_0^{2}/512\pi^{2}Q_m^{2}F)^{1/6}$.
At the throat, $\LL(r_0)=\Lambda/(8\pi)+1/(16\pi r_0^{2})$ and
$\LL_F(r_0)=16\pi Q_m^{2}r_0^{2}>0$ for any $Q_m\neq0$, so the null energy
condition is violated there by the third case of \eqref{eq:NECbranches},
as required; and since $\LL_F(r)>0$ for every $r>0$, it stays violated
throughout the exterior, with no finite shell as in the analogous example of Sec.~\ref{sec:el-reconstruction}. Domain and flare-out hold trivially for $b(r)=r_0$, for any
$\Lambda$: as in Sec.~\ref{sec:el-reconstruction}, $\Lambda$ is fully absorbed into $\LL$ and
never required to vanish here.

Equation \eqref{eq:magnetic-reconstruction} also recovers the static family
of Ref.~\cite{MAH}, which fixes $\LL(F)\propto\sqrt{F}$ (their $F$ being
$F_{\mu\nu}F^{\mu\nu}=4F$ in the present convention) with no additive
constant. Every shape function they consider is a special case of this
subsection: their MTtW ansatz is $b(r)=b_0^{2}/r$, and their generalisation
$b(r)=b_0^{\mu+1}/r^{\mu}$; both, and any other $b(r)$ satisfying
flare-out, are covered at once by \eqref{eq:magnetic-B}, with $\Lambda$
left arbitrary rather than fixed to zero as in their analysis. What
\eqref{eq:magnetic-B} adds is therefore not a new theory but the full space
of geometries that theory supports, of which Ref.~\cite{MAH} explored two
one-parameter families.

Comparing \eqref{eq:magnetic-B} directly against Ref.~\cite{MAH} uncovers a
further, independent discrepancy. Their field equations (their Eqs.~19-20)
give $G^{t}_{t}=-(b'-b/r)/(2r^{2})=T^t_t$ with
$T^t_t=(\alpha/\sqrt2)B\sqrt{1-b/r}/r$; solving this pair directly for $B(r)$
gives
\begin{equation}
B(r) = \frac{\sqrt2\big[b(r)-rb'(r)\big]}{2\alpha\,r^{3/2}\sqrt{r-b(r)}} ,
\end{equation}
which, converted to the orthonormal frame used here, has exactly the same
$r^{-3}$ dependence as \eqref{eq:magnetic-B}. Their published relation
(their Eq.~22), however, reads
\begin{equation}
B(r) = -\frac{\sqrt2\big(rb'-b\big)}{2\alpha\,r^{3}\sqrt{1-b/r}} ,
\end{equation}
exactly a factor of $r$ smaller than what their own Eqs.~19-20 imply; the
same extra factor propagates into their explicit examples (their Eqs.~23
and 34). Equation \eqref{eq:magnetic-B} is, up to normalisation of the
coupling constant, the corrected relation.

\subsection{The single power law, $B$ forced constant}\label{sec:mag-power}
Section~\ref{sec:mag-reconstruction} fixed nothing about $\LL(F)$ beyond
\eqref{eq:magnetic-reconstruction}; taking the opposite direction, fix
$\LL(F)=\gamma F^{k}$ for some power $k$ instead, and ask what
\eqref{eq:gauss}(iii) allows. With $F=B^{2}/2$,
\begin{equation}
B\,\LL_F\big(B^{2}/2\big) = 2\gamma k\Big(\frac{B^{2}}{2}\Big)^{k-1}\frac{B}{2}
= Q_m ,
\end{equation}
an equation in $B$ alone, with no $r$-dependence at all. For $k=1/2$ this
collapses to a condition on the constants, $\gamma=\sqrt2\,Q_m$, leaving
$B$ free to be any function of $r$, and then recovering Sec.~\ref{sec:mag-reconstruction}. For every other $k$, by contrast, it fixes $B$ to a constant, since $B$ is independent of $r$,
\begin{equation}
B(r) \equiv B_0 = \Big(\frac{2^{k-1}Q_m}{\gamma k}\Big)^{1/(2k-1)} .
\label{eq:magnetic-B0}
\end{equation}
A fixed power therefore does not leave the field undetermined in the same
way that the pure-power models of Sec.~\ref{sec:el-forward} left $b(r)$
determined up to a single constant. Instead, it completely fixes the radial
dependence of the field, leaving $b(r)$ as the remaining free function, with
\eqref{eq:magnetic-B} interpreted as an ODE for $b$ at fixed $B=B_0$,
\begin{equation}
b(r) = C_1 r - 8\pi Q_m B_0\,r^{3} ,
\label{eq:magnetic-b-other-k}
\end{equation}
with $C_1$ a single integration constant, the same pattern
found for a fixed power in Sec.~\ref{sec:el-forward}, and for a generic
redshift function in Sec.~\ref{sec:el-generic}, now appearing for a fixed
power in this branch instead.
The $tt$ component of the field equations, $\LL=\Lambda/(8\pi)+Q_mB$ in \eqref{eq:magnetic-LLF}, is then no longer automatic. With $B=B_0$ and $\LL=\gamma(B_0^{2}/2)^{k}$ it fixes the cosmological constant,
\begin{equation}
\Lambda=8\pi\gamma\,(1-2k)\Big(\frac{B_0^{2}}{2}\Big)^{k}=\frac{4\pi(1-2k)}{k}\,Q_mB_0 ,
\label{eq:magnetic-Lambda-k}
\end{equation}
which is negative for every $k>1/2$. As in the azimuthal branch, fixing $\LL(F)$ to a single power forces $\Lambda$ to a definite value.

\emph{Example.} Take $k=2$, so $B_0=(Q_m/\gamma)^{1/3}$ from
\eqref{eq:magnetic-B0}, and $\Lambda=-6\pi Q_mB_0$ from \eqref{eq:magnetic-Lambda-k}. Fixing the throat at $r_0$ in
\eqref{eq:magnetic-b-other-k} gives $C_1=1+8\pi Q_mB_0r_0^{2}$ and
\begin{equation}
b(r) = r + 8\pi Q_mB_0\,r_0^{2}r - 8\pi Q_mB_0\,r^{3}.
\end{equation}
Direct evaluation confirms $b(r_0)=r_0$, and
\begin{gather}
b'(r_0) = 1-16\pi Q_mB_0r_0^{2} < 1 , \\
1-\frac{b(r)}{r} = 8\pi Q_mB_0\big(r^{2}-r_0^{2}\big) > 0 \quad (r>r_0) ,
\end{gather}
both automatic for any $Q_m,B_0>0$, exactly as in
the fixed-power examples of Sec.~\ref{sec:el-forward}. Although the
domain condition above holds for every $r>r_0$, the embedding function,
$z'(r)=\sqrt{b(r)/(r-b(r))}$, is not defined everywhere: since
$r-b(r)=8\pi Q_mB_0\,r(r^{2}-r_0^{2})>0$ throughout, it is $b(r)$ itself
that decides the sign, and $b(r)$ changes sign at
\begin{equation}
r_{\max} = \sqrt{r_0^{2}+\frac{1}{8\pi Q_mB_0}} ,
\label{eq:magnetic-k2-rmax}
\end{equation}
beyond which $z'(r)$ becomes imaginary and the embedding in $\mathbb E^{3}$
ceases to apply, even though the wormhole itself, and the field sourcing
it, remain perfectly well defined for every $r>r_0$. Unlike the domain
condition itself, $z(r)$ has no elementary closed form; \eqref{eq:magnetic-k2-rmax}
reduces the problem to a single quadrature,
\begin{equation}
z(r) = \int_{r_0}^{r} \sqrt{\frac{b(r')}{r'-b(r')}}\;\rd r' ,
\label{eq:magnetic-k2-z}
\end{equation}
which we evaluate numerically. Figure~\ref{fig:conformal-embedding-k2}
shows the result.

\begin{figure*}[t]
\centering
\includegraphics[width=\linewidth]{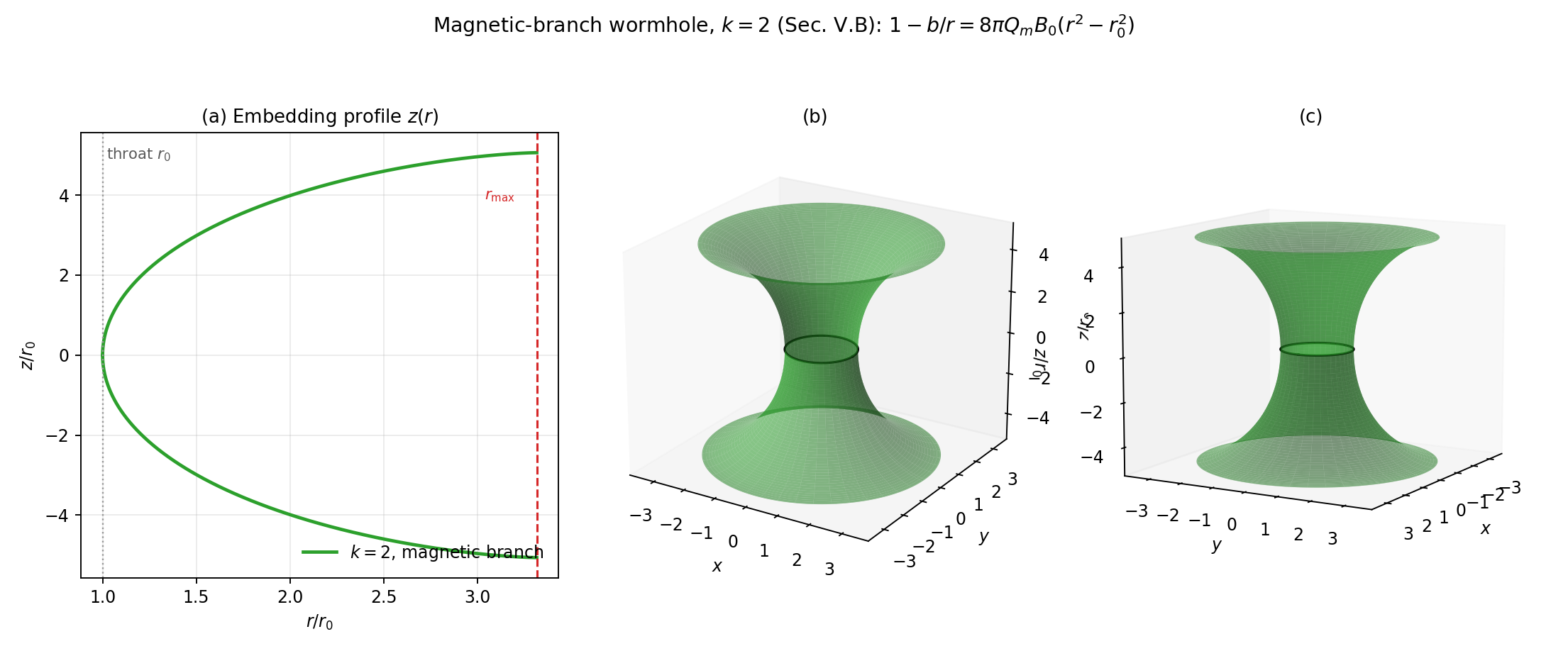}
\caption{Embedding diagram for the magnetic-branch
wormhole with $k=2$ (Sec.~\ref{sec:mag-power}), obtained by numerically
integrating \eqref{eq:magnetic-k2-z}. Panel (a) shows the profile $z(r)$,
bounded between the throat $r_0$ and $r_{\max}$ from
\eqref{eq:magnetic-k2-rmax} (dashed); panels (b)--(c) show the
corresponding surface of revolution from two viewpoints.}
\label{fig:conformal-embedding-k2}
\end{figure*}

Because $B=B_0$ is constant, so is $F=B_0^{2}/2$, and with it $\LL_F=Q_m/B_0$: the null
energy condition, violated wherever $\LL_F>0$ by the third case of
\eqref{eq:NECbranches}, is therefore violated uniformly, at every radius
rather than in a finite shell, a genuinely new solution, sourced by a
Lagrangian Ref.~\cite{MAH} never considered, with a shape function found
nowhere in the earlier literature on this branch either.

\subsection{The generic branch, $\Phi\not\equiv\text{const}$}\label{sec:mag-generic}

Outside the scope of Theorem~\ref{thm:magnetic-rigidity}, \eqref{eq:magnetic-master} is a genuine second-order linear ODE for $b(r)$, for any given $\Phi(r)$ with $\Phi'\not\equiv0$. Its general solution is
\begin{equation}
b(r) = r + \frac{c_1\,r^{3}}{e^{2\Phi(r)}\big[\Phi'(r)\big]^{2}} ,
\label{eq:magnetic-b-generic}
\end{equation}
with $c_1$ a single integration constant. As in the azimuthal branch, $b(r)$ is here fixed up to one constant, not free, and \eqref{eq:magnetic-b-generic} degenerates back into
the freedom of Theorem~\ref{thm:magnetic-rigidity} only where $\Phi'(r)=0$.

\begin{proposition}[No throat in the generic magnetic branch]
\label{prop:magnetic-generic-nothroat}
If $c_1\neq0$ and $\Phi'(r)\neq0$ throughout the domain, then $b(r)\neq r$
for every $r$: no throat exists.
\end{proposition}

\begin{proof}
From \eqref{eq:magnetic-b-generic}, $b(r)-r=c_1r^{3}/\big(e^{2\Phi}(\Phi')^{2}\big)$.
Since $e^{2\Phi}>0$ and $(\Phi')^{2}>0$ wherever $\Phi'\neq0$, this
expression has the sign of $c_1$ throughout and never vanishes.
\end{proof}

A throat can therefore only appear at a point $r_0$ where $(\Phi')^2$ itself diverges, that is, where $\Phi'(r_0)\to\infty$. This is the same mechanism found in Sec.~\ref{sec:el-generic}, now for the magnetic branch.

\emph{The case $b(r)=r_0$ constant.} Setting $b(r)=r_0$ directly in
\eqref{eq:magnetic-master}, with $b'=0$, gives a second-order ODE for
$\Phi(r)$ alone. Substituting $u=\Phi'$ linearises it to $y''=Qy'$ with
$y\equiv$ (up to the substitution $u=y'/y$) an auxiliary function and
$Q(r)=3/(2r)-1/[2(r-r_0)]$; integrating twice gives, after simplification,
\begin{multline}
e^{\Phi(r)} = \frac{C_1}{4}\Big[(2r+3r_0)\sqrt{r(r-r_0)} \\
+3r_0^{2}\ln\Big(\frac{\sqrt r+\sqrt{r-r_0}}{\sqrt{r_0}}\Big)\Big] + C_2 ,
\qquad r\geq r_0 ,
\label{eq:magnetic-Phi-const-b}
\end{multline}
with $C_1,C_2$ integration constants. The construction is structurally parallel to the expression in $\eqref{eq:Phi-const-b}$ of Sec.~\ref{sec:el-generic}, but the result is not: alongside the same square-root term, $\eqref{eq:magnetic-Phi-const-b}$ carries an additional
logarithmic term that $\eqref{eq:Phi-const-b}$ does not have, so unlike
that purely algebraic solution, $e^{\Phi}$ here is a transcendental
function of $r$.

For $C_1,C_2>0$, the throat satisfies $e^{\Phi(r_0)}=C_2$, which is finite and positive. For $r>r_0$, every term inside the bracket is positive and
increasing, so $e^{\Phi(r)}\geq C_2>0$ throughout, with no horizon
anywhere. As in Sec.~\ref{sec:el-generic}, $\Phi'(r_0)\to\infty$, a
coordinate artifact rather than a curvature singularity: switching to the
proper radial distance $\ell$, direct computation gives
$\rd\Phi/\rd\ell|_{r_0}=C_1r_0/C_2$, finite.

The matter content follows from \eqref{eq:coupled-tt} and
\eqref{eq:coupled-rr} directly. Equation \eqref{eq:coupled-tt} never
involves $\Phi$, so $\LL(r_0)=\Lambda/(8\pi)+1/(16\pi r_0^{2})$, exactly as
in Secs.~\ref{sec:mag-reconstruction}--B. Evaluating $B^{2}\LL_F$ from
\eqref{eq:coupled-rr} at the throat, the divergent pieces from $\Phi'(r_0)$
cancel against the vanishing $(1-b/r_0)$ prefactor, leaving
\begin{equation}
B^{2}(r_0)\,\LL_F(r_0) = \frac{1}{16\pi r_0^{2}} ,
\end{equation}
finite, and identical to the throat value found in
Secs.~\ref{sec:mag-reconstruction}--B despite the different redshift function: the
null energy condition is violated there, as required, by exactly the same
amount regardless of which of this branch's constructions produced the
throat. This is a genuine, traversable wormhole with $b(r)=r_0$ constant,
distinct from the example of Sec.~\ref{sec:mag-reconstruction} with the same shape
function, since here $\Phi\neq\text{const}$.

The Lagrangian again fails to admit a closed form here, for the
same reason as in Sec.~\ref{sec:el-generic}: $F(r)$ is no longer a pure
power of $r$ once $\Phi(r)$ is this involved, and \eqref{eq:magnetic-Phi-const-b}
has no elementary inverse $r(F)$ for $C_1\neq0$.

\section{Coulomb-like nonlinear traversable wormholes}\label{sec:examples}

The classification of Secs.~\ref{sec:electric} and~\ref{sec:magnetic} holds
for an unrestricted $\LL(F)$. It is worth applying it to a Lagrangian
singled out on physical grounds rather than chosen for convenience. For
$\LL(F)\propto F^{k}$ the trace of the stress tensor is
\begin{equation}
T^{\hd\mu}{}_{\hd\mu} = D\,\LL - 4F\,\LL_F = \gamma F^{k}(D-4k) ,
\end{equation}
where $D$ is the total spacetime dimension; it vanishes identically only
for $k=D/4$. In $(3+1)$-dimensions this singles out linear Maxwell theory,
$k=1$. In $(2+1)$-dimensions, $D=3$, the analogous conformally invariant
power is $k=3/4$, the genuine three-dimensional counterpart of Maxwell's
own conformal invariance, rather than an arbitrary point on the family of
powers treated in Secs.~\ref{sec:el-forward} and~\ref{sec:mag-power}.

This same power was singled out, on the same trace-free grounds, by Ref.~\cite{CataldoCruzDelCampoGarcia}, in the radial electric branch. There the resulting field was found to be
\begin{equation}
E(r) = \frac{q}{r^{2}} ,
\end{equation}
the genuine Coulomb law of $(3+1)$-dimensional Minkowski space, realised inside a $(2+1)$-dimensional theory purely as a consequence of demanding a traceless stress tensor, which is the origin of the name we adopt for this Lagrangian throughout this section. That solution sources a charged (anti-)de Sitter \emph{black hole}, with
\begin{equation}
e^{2\Phi(r)} = 1-\frac{b(r)}{r} = -M-\Lambda r^{2}+\frac{4q^{2}}{3r} ,
\end{equation}
horizons located at the roots of a cubic, and finite quasilocal mass; it
is not a traversable wormhole, in agreement with Theorem~\ref{thm:radial}:
the radial electric branch admits no throat for \emph{any} $\LL(F)$,
Coulomb-like or otherwise. Consistently, that solution's own energy condition, $-(\LL+E^{2}\LL_F)=q^{2}/(12\pi r^{3})\geq0$, satisfies the weak energy condition throughout, with the opposite sign to that required by a throat. The coupling in Ref.~\cite{CataldoCruzDelCampoGarcia} therefore carries the opposite sign from the one we use below for exactly this reason.

We now show that the same conformally invariant, Coulomb-like theory,
applied instead to the azimuthal electric field or to the magnetic field
of Secs.~\ref{sec:electric} and~\ref{sec:magnetic}, does support a genuine
throat in both cases: it is the branch, not the Lagrangian, that decides
between a black hole and a wormhole.

\subsection{Azimuthal branch}

Take $\LL(F)=-\alpha(-2F)^{3/4}$ in \eqref{eq:general-k-b}--\eqref{eq:general-k-Lambda}.
With $F(r)=-Q^{2}/(2r^{4})$ from \eqref{eq:azimuthal-field}, the exponent
of the field-sourced term is $3-4k=0$: unlike every other power, the
field's contribution to $b(r)$ is here a pure constant, not a new power of
$r$,
\begin{equation}
b(r) = r + \Lambda r^{3} + 8\pi\alpha Q^{3/2} ,
\label{eq:conformal-az-b}
\end{equation}
with
\begin{equation}
\LL(r) = -\frac{\alpha Q^{3/2}}{r^{3}} , \qquad
\LL_F(r) = \frac{3\alpha r}{2\sqrt Q} .
\label{eq:conformal-az-L}
\end{equation}
Fixing the throat at $r_0$ in \eqref{eq:conformal-az-b} gives
\begin{equation}
\Lambda = -\frac{8\pi\alpha Q^{3/2}}{r_0^{3}} ,
\label{eq:conformal-az-Lambda}
\end{equation}
in agreement with \eqref{eq:general-k-Lambda} at $k=3/4$, and
\begin{equation}
1-\frac{b(r)}{r} = \frac{8\pi\alpha Q^{3/2}}{r\,r_0^{3}}\,(r-r_0)\big(r^{2}+rr_0+r_0^{2}\big) > 0, \label{eq:conformal-az-domain}
\end{equation}
for $r>r_0$ and
\begin{equation}
b'(r_0) = 1-\frac{24\pi\alpha Q^{3/2}}{r_0} < 1 ,
\end{equation}
both automatic for any $\alpha,Q>0$, exactly as in
the general power-law case of Sec.~\ref{sec:el-forward}. Explicitly,
the resulting spacetime is
\begin{eqnarray}
ds^{2} = &-& r^{2}\,\rd t^{2}
+ \frac{r\,r_0^{3}\,\rd r^{2}}{8\pi\alpha Q^{3/2}(r-r_0)(r^{2}+rr_0+r_0^{2})} \nonumber \\
&+& r^{2}\,\rd\varphi^{2} ,
\label{eq:conformal-az-metric}
\end{eqnarray}
with $\mathcal E(r)=Q/r^{2}$ from \eqref{eq:azimuthal-field} at $e^{\Phi}=r$.
Since $\LL_F(r_0)=3\alpha r_0/(2\sqrt Q)>0$, and more generally $\LL_F(r)>0$ for
every $r>0$, the null energy condition is violated throughout the
exterior rather than confined to a shell, as for any single power with
$\alpha,k>0$: the same field content as
Ref.~\cite{CataldoCruzDelCampoGarcia}, sourcing a throat instead of a horizon once
it is placed in the azimuthal rather than the radial branch.

As in the fixed-power example of Sec.~\ref{sec:mag-power}, the
domain condition \eqref{eq:conformal-az-domain} holds for every $r>r_0$,
but the embedding function $z'(r)=\sqrt{b(r)/(r-b(r))}$ does not: $b(r)$
changes sign at the positive root $r_{\max}$ of the cubic
\begin{equation}
r_{\max} + \Lambda r_{\max}^{3} + 8\pi\alpha Q^{3/2} = 0 ,
\label{eq:conformal-az-rmax}
\end{equation}
beyond which $z'(r)$ becomes imaginary, even though the wormhole itself
remains well defined for every $r>r_0$. Unlike the magnetic branch's
\eqref{eq:magnetic-k2-rmax}, the cubic \eqref{eq:conformal-az-rmax} has no
simple closed-form root and $r_{\max}$ is found numerically, together with
the embedding
\begin{equation}
z(r) = \int_{r_0}^{r} \sqrt{\frac{b(r')}{r'-b(r')}}\;\rd r' ,
\label{eq:conformal-az-z}
\end{equation}
shown in Figure~\ref{fig:conformal-embedding-az}.

\begin{figure*}[t]
\centering
\includegraphics[width=\linewidth]{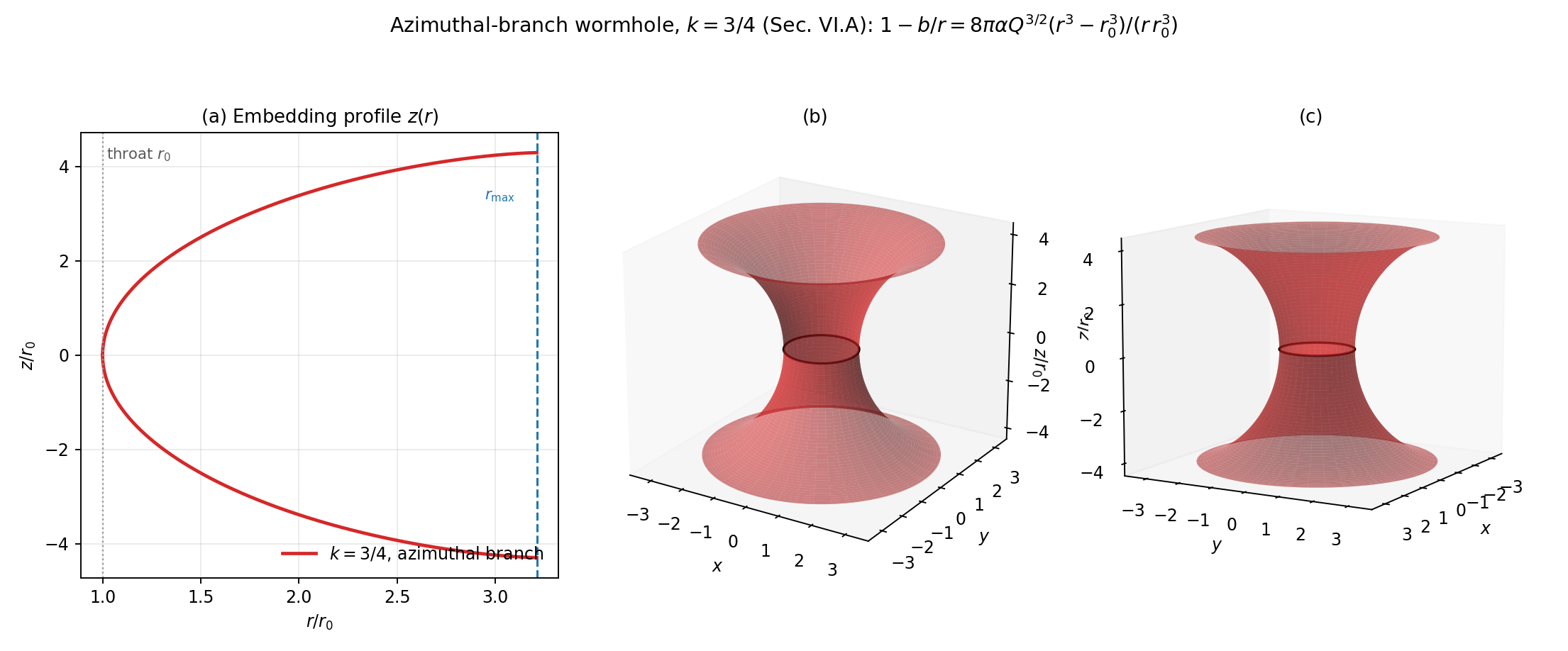}
\caption{Embedding diagram for the azimuthal-branch
wormhole with $k=3/4$ (Sec.~\ref{sec:examples}.A), obtained by numerically
integrating \eqref{eq:conformal-az-z}. Panel (a) shows the profile $z(r)$,
bounded between the throat $r_0$ and $r_{\max}$ from
\eqref{eq:conformal-az-rmax} (dashed); panels (b)--(c) show the
corresponding surface of revolution from two viewpoints.}
\label{fig:conformal-embedding-az}
\end{figure*}

\subsection{Magnetic branch}

Take $\LL(F)=\gamma F^{3/4}$ in \eqref{eq:magnetic-B0}. This fixes
\begin{equation}
B \equiv B_0 = \frac{8\sqrt2\,Q_m^{2}}{9\gamma^{2}} ,
\label{eq:conformal-mag-B0}
\end{equation}
and \eqref{eq:magnetic-b-other-k}, with the throat fixed at $r_0$, gives
\begin{equation}
b(r) = r + 8\pi Q_mB_0\,r_0^{2}r - 8\pi Q_mB_0\,r^{3} .
\label{eq:conformal-mag-b}
\end{equation}
Direct evaluation confirms $b(r_0)=r_0$, and
\begin{eqnarray}
1-\frac{b(r)}{r} &=& 8\pi Q_mB_0\big(r^{2}-r_0^{2}\big) > 0 \quad (r>r_0), \nonumber \\
b'(r_0) &=& 1-16\pi Q_mB_0r_0^{2} < 1 ,
\end{eqnarray}
both automatic for any $Q_m,B_0>0$. Explicitly, the
resulting spacetime is
\begin{equation}
ds^{2} = -\rd t^{2} + \frac{\rd r^{2}}{8\pi Q_mB_0(r^{2}-r_0^{2})}
+ r^{2}\,\rd\varphi^{2} ,
\label{eq:conformal-mag-metric}
\end{equation}
with $B(r)\equiv B_0$ from \eqref{eq:conformal-mag-B0} and $\Phi\equiv0$, and with the cosmological constant fixed by \eqref{eq:magnetic-Lambda-k} at $\Lambda=-(8\pi/3)\,Q_mB_0$.
Because $B\equiv B_0$ is constant here,
so is $\LL_F(r_0)=Q_m/B_0>0$: the null energy condition is again violated
uniformly, at every radius, rather than in a finite shell.

\subsection{Comparison}

The same conformally invariant, Coulomb-like Lagrangian therefore produces
three qualitatively different outcomes depending only on which field
component sources it: a charged black hole in the radial electric branch
(Ref.~\cite{CataldoCruzDelCampoGarcia}), and a traversable wormhole in each of the
remaining two. The azimuthal and magnetic wormholes share the same
asymptotic structure, $1-b/r\sim|\Lambda|\,r^{2}$ and
$1-b/r\sim8\pi Q_mB_0\,r^{2}$ respectively as $r\to\infty$, both AdS-like;
they differ near the throat, where the azimuthal shape function grows
cubically in $(r-r_0)$ while the magnetic one grows quadratically.
Figure~\ref{fig:conformal-embedding} shows the embedding diagrams for
both, at matched throat radius $r_0$ and comparable amplitude of the
coupling constants. In contrast to the $k=1/2$ Lagrangians of
Refs.~\cite{CanateBreton,MHG}, whose choice carries no independent
physical justification, the power used throughout this section is fixed
uniquely by the requirement of a traceless stress tensor, and its
consequences depend entirely on which of the three branches of
Proposition~\ref{prop:branches} it is made to source.

\begin{figure*}[t]
\centering
\includegraphics[width=\linewidth]{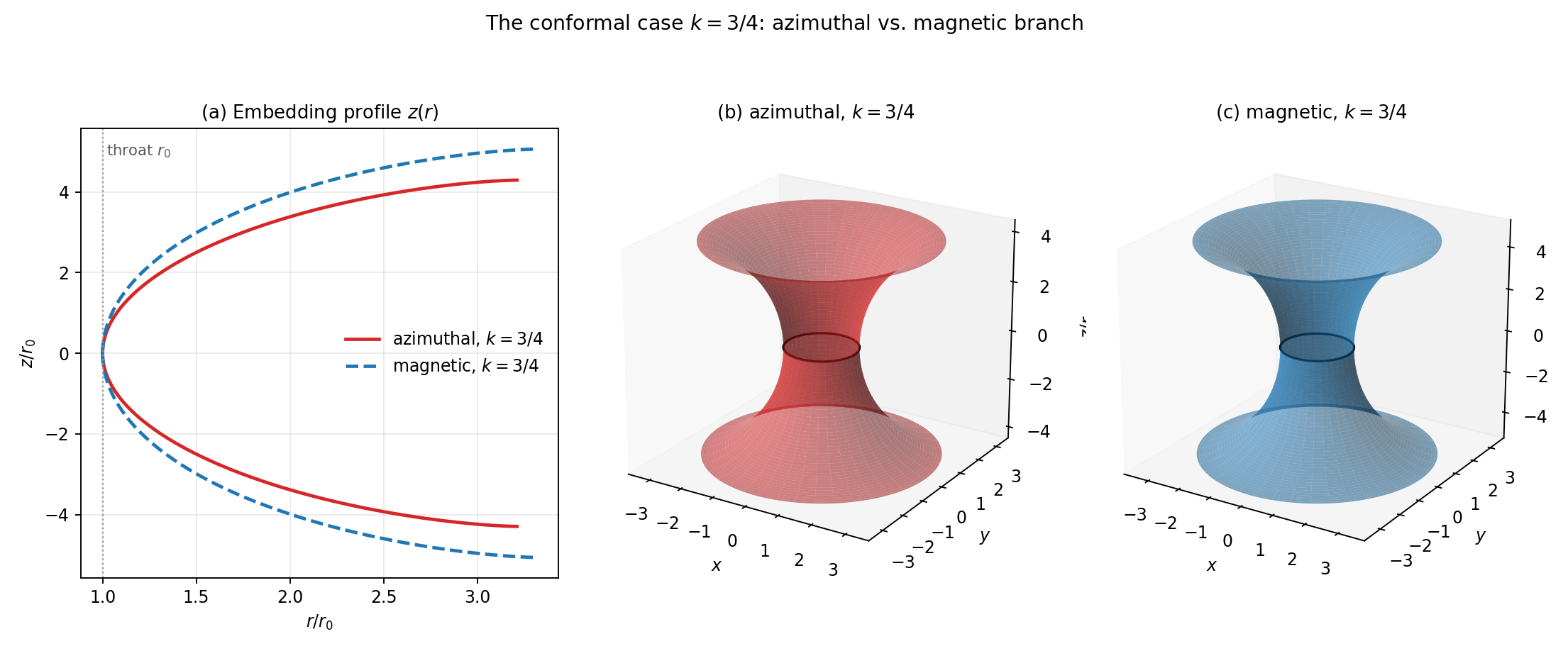}
\caption{Embedding diagrams for the conformal, Coulomb-like case $k=3/4$,
azimuthal branch (red solid curve, Sec.~\ref{sec:examples}.A) and magnetic
branch (blue dashed curve, Sec.~\ref{sec:examples}.B), at matched throat
radius $r_0$ and comparable coupling amplitude. Panel (a) shows the profile
$z(r)$; panels (b)--(c) show the corresponding surfaces of revolution. Both
are asymptotically AdS and are shown out to the point where $b(r)=r$, beyond
which the embedding in $\mathbb{E}^{3}$ ceases to apply. The radial
electric branch of the same theory has no such diagram: it is a black
hole, not a wormhole (Ref.~\cite{CataldoCruzDelCampoGarcia}).}
\label{fig:conformal-embedding}
\end{figure*}

Another nonlinear electrodynamics of interest in $(2+1)$ dimensions is Born-Infeld theory. Its thermodynamics were worked out, for a charged
circularly symmetric black hole in the radial electric branch, by Ref.~\cite{Cataldo:1999wr}. That solution is a genuine complement to
the Coulomb-like example of Ref.~\cite{CataldoCruzDelCampoGarcia} discussed above: with the standard normalisation $\LL(F)=-\frac{\beta^2}{4\pi}\big(\sqrt{1+2F/\beta^2}-1\big)$
recovering Maxwell at weak field, $\LL_F(F)=-\beta/\big(4\pi\sqrt{2F+\beta^2}\big)$ is negative for every admissible $F$, in either branch, not only the
radial one; the null energy condition is never violated anywhere, so Born-Infeld admits no throat in the azimuthal or magnetic branch either,
consistent with the solution of Ref.~\cite{Cataldo:1999wr} itself being a genuine charged black hole rather than a wormhole.

\section{Conclusions} \label{sec:conclusions}

We have given a complete classification of static, circularly symmetric
traversable wormholes in $(2+1)$-dimensional Einstein gravity coupled to
an arbitrary nonlinear electrodynamics $\LL(F)$. Proposition~\ref{prop:branches}
shows that the electromagnetic field compatible with the symmetry
occupies exactly one of three mutually exclusive configurations: radial
electric, azimuthal electric, or magnetic. Theorem~\ref{thm:radial}
excludes the first of these entirely, for any $\LL(F)$ and any $\Lambda$: the
radial branch saturates the null energy condition identically and any
attempted throat is forced to coincide with an event horizon, a
purely geometric obstruction independent of the sign of $\LL_F$.

In each of the remaining two branches, Theorems~\ref{thm:rigidity}
and~\ref{thm:magnetic-rigidity} show that the field equations admit
exactly two regimes, and no third. In the first, the redshift
function is fixed ($e^{\Phi}=Cr$ in the azimuthal branch,
$\Phi=\text{const}$ in the magnetic one) and the shape function $b(r)$ is
completely free, with $\Lambda$ absorbed without constraint into
$\LL(r)$; we reconstructed $\LL(F)$ in closed form for arbitrary $b(r)$ in
both branches, with the field, $\LL$, and $\LL_F$ finite at the throat. In
the generic regime, the roles invert: $b(r)$ is fixed up to a single
integration constant by Eq.~\eqref{eq:b-generic} (azimuthal) or
\eqref{eq:magnetic-b-generic} (magnetic), and
Propositions~\ref{prop:generic-nothroat} and~\ref{prop:magnetic-generic-nothroat}
show that no throat exists there except at the boundary of that formula's
validity, where the redshift function itself becomes singular in a
controlled, non-physical way; we exhibited and fully verified the
resulting exceptional solutions in closed form in both branches. Fixing
$\LL(F)$ to a single power $|F|^k$, for any $k$ and not only the $k=1/2$
of the prior literature, was shown to force $\Lambda\neq0$ and to leave
$b(r)$ determined up to a finite number of integration constants in the azimuthal branch, and to
force the magnetic field itself to a constant in the magnetic branch: the freedom of the fixed-redshift regime requires $\LL(F)$ to be genuinely unrestricted, not merely a free exponent. One member of the azimuthal single-power family reproduces the static BTZ mass function exactly, the same geometry that is a black hole in vacuum becoming a traversable wormhole once sourced by the nonlinear field instead of $\Lambda$ alone.

Along the way we corrected two results in the prior literature. The
non-existence proof of Ref.~\cite{ArellanoLobo} was shown to rest on an
ansatz that omits the azimuthal electric field entirely, and its own
radial-branch argument to contain a sign/degree slip that does not affect
its conclusion; and the shape-function--field relation of
Ref.~\cite{MAH} was shown to carry a spurious factor of $r$, corrected in
Eq.~\eqref{eq:magnetic-B}.

Section~\ref{sec:examples} applied the classification to a Lagrangian
singled out on physical rather than illustrative grounds: the unique
power-Maxwell theory with a traceless stress tensor in $(2+1)$ dimensions,
$k=3/4$. This is the same Lagrangian already known to source a
Coulomb-like charged black hole in the radial branch~\cite{CataldoCruzDelCampoGarcia};
we showed that it sources genuine traversable wormholes instead in both
the azimuthal and magnetic branches, with the null energy condition
violated throughout the exterior in the former and uniformly at every
radius in the latter, as holds for any single
power with $\alpha,k>0$, in contrast to the finite-shell confinement found
for the reconstructed Lagrangian of Sec.~\ref{sec:el-reconstruction}. The three outcomes
of the same theory, one per branch, are the most direct illustration of
this paper's central claim: it is the electromagnetic configuration, not
the choice of Lagrangian, that decides between a horizon and a throat.
Born-Infeld electrodynamics~\cite{Cataldo:1999wr} makes the same point
from the opposite direction: with the normalisation that recovers Maxwell
at weak field, $\LL_F(F)<0$ for every admissible $F$ in every branch, so
the null energy condition is never violated and no throat exists anywhere
in this theory, consistent with its own radial-branch solution being a
genuine black hole. The classification's freedom is real, but it is not
unconditional.

The methods used throughout, direct integration of the coupled field equations in the areal-radius gauge, apply without modification to any power-Maxwell or more general $\LL(F)$ theory in this symmetry class, and the same dichotomy between fixed and generic redshift function is likely to recur in other low-dimensional settings where the electromagnetic field admits more than one inequivalent static configuration.

The present classification is restricted to the static sector, and its natural extension is the stationary one. There the branch structure is already different, since by the theorem of Ca\~nate and Bret\'on~\cite{CanateBreton} a radial electric and a magnetic field can coexist, so the three static branches of Proposition~\ref{prop:branches} reduce to two, and the field equations acquire an additional component through the shift function. The five-parameter wormhole family of Ref.~\cite{CanateBreton} belongs to that sector, and also requires $\Lambda<0$ for its $k=1/2$ Lagrangian. Whether that requirement is, as in the static case, an artefact of the choice $k=1/2$ is an open question; the direct-integration method used here, extended to a nonvanishing shift function, is the natural tool to settle it, and we leave it for future work.

\end{document}